\documentclass[envcountsect,orivec]{eptcs}
\providecommand{\event}{SNR 2021} 
\usepackage{etex} 
\usepackage[]{graphicx}
\usepackage{verbatim}

\newif\ifignore 
\ignorefalse
\newcommand{\auxproof}[1]{
  \ifignore\mbox{}\newline
  \textbf{BEGIN: AUX-PROOF} \dotfill\newline
  {#1}\mbox{}\newline
  \textbf{END: AUX-PROOF}\dotfill\newline
  \fi}

\usepackage{amsthm}
\usepackage{fancybox,amssymb,amstext,amsmath,stmaryrd,wasysym,cite,proof,mathtools,multirow,stackrel}
\SetSymbolFont{stmry}{bold}{U}{stmry}{m}{n}
\SetSymbolFont{wasy}{bold}{U}{wasy}{m}{n}

\usepackage{algorithm}
\usepackage[noend]{algpseudocode} 
\usepackage{etoolbox}

\makeatletter
\newcommand*{\algrule}[1][\algorithmicindent]{%
   \makebox[#1][l]{%
       \hspace*{.2em}
       \vrule height .75\baselineskip depth .25\baselineskip
   }
}

\newcount\ALG@printindent@tempcnta
\def\ALG@printindent{%
    \ifnum \theALG@nested>0
    \ifx\ALG@text\ALG@x@notext
    \else
    \unskip
    \ALG@printindent@tempcnta=1
    \loop
    \algrule[\csname ALG@ind@\the\ALG@printindent@tempcnta\endcsname]%
    \advance \ALG@printindent@tempcnta 1
    \ifnum \ALG@printindent@tempcnta<\numexpr\theALG@nested+1\relax
    \repeat
    \fi
    \fi
}
\patchcmd{\ALG@doentity}{\noindent\hskip\ALG@tlm}{\ALG@printindent}{}{\errmessage{failed to patch}}
\patchcmd{\ALG@doentity}{\item[]\nointerlineskip}{}{}{} 
\makeatother

\usepackage[all,pdftex]{xy}
\CompileMatrices 
\xyoption{v2}
\xyoption{curve}
\xyoption{2cell}
\SelectTips{cm}{}  
\UseAllTwocells

\usepackage{wrapfig}
\theoremstyle{plain}
\newtheorem{mytheorem}{Theorem}[section]

\newtheorem{myproposition}[mytheorem]{Proposition}

\theoremstyle{definition}

\newtheorem{mydefinition}[mytheorem]{Definition}

\def\myqed{\par\qedhere}

\usepackage{xcolor}
\definecolor{dgreen}{rgb}{0, .6, 0}

\usepackage{pifont}
\newcommand{\bu}{\mathbf{u}}
\newcommand{\bv}{\mathbf{v}}
\newcommand{\bw}{\mathbf{w}}
\newcommand{\N}{{\mathbb{N}}}

\newcommand{\R}{{\mathbb{R}}}
\newcommand{\Rm}{\R^{m}}
\newcommand{\M}{\mathcal{M}}
\newcommand{\Rpos}{\R_{>0}}

\newcommand{\Var}{\mathbf{Var}}
\newcommand{\Falsify}{\mathsf{Falsify}}
\newcommand{\ttrue}{\mathrm{t{\kern-1.5pt}t}}
\newcommand{\ffalse}{\mathrm{f{\kern-1.5pt}f}}

\newcommand{\Table}[1]{Table~\ref{tab:#1}}
\newcommand{\Def}[1]{Def.~\ref{def:#1}}

\newcommand{\Fig}[1]{Fig~\ref{fig:#1}}

\newcommand{\place}{\underline{\phantom{n}}\,}

\newcommand{\STL}{\textbf{STL}}

\newcommand{\UntilOp}[1]{\mathbin{\mathcal{U}_{#1}}}

\newcommand{\DiaOp}[1]{\Diamond_{#1}}
\newcommand{\BoxOp}[1]{\square_{#1}}

\newcommand{\Robust}[2]{{ \llbracket #1, #2 \rrbracket}}
\newcommand{\sem}[1]{\llbracket #1 \rrbracket} 
\newcommand{\Defeq}{:=}

\newcommand{\Vee}[1]{{{\bigsqcup_{#1}}}}
\newcommand{\Wedge}[1]{{{\bigsqcap_{#1}}}}

\newcommand{\Rnn}{\R_{\ge 0}}

\makeatletter
\newcommand{\figcaption}[1]{\def\@captype{figure}\caption{#1}}
\newcommand{\tblcaption}[1]{\def\@captype{table}\caption{#1}}
\makeatother

\newif\iftikzgnuplot
\tikzgnuplottrue 
\usepackage{gnuplot-lua-tikz}
\usepackage{pgfplotstable} 
\pgfplotsset{compat=1.12}
\usepackage{booktabs}
\usepackage{tabularx}

\usepackage{arydshln}
\usepackage{hyperref}

\usepackage{enumitem}
\setlistdepth{20}
\renewlist{itemize}{itemize}{20}
\setlist[itemize]{label=\textbullet}

\DeclareMathOperator*{\argmin}{arg\,min}

\usepackage{tikz}
\usetikzlibrary{automata,positioning,matrix,shapes.callouts}
\tikzset{
region/.style={
rectangle,
rounded corners,
draw=black,very thick
},
accepting/.style={double distance=2pt}
}

\makeatletter
\newcommand{\Biggg}{\bBigg@{4}}
\newcommand{\Bigggg}{\bBigg@{5}}
\newcommand{\Biggggg}{\bBigg@{6}}
\makeatother

\title{Time-Staging Enhancement of\\  Hybrid System Falsification
     \thanks{Presented at SNR'2018. G. Ernst and Z. Zhang were then at the National Institute of Informatics.}}
\def\titlerunning{Time-Staging Enhancement of Hybrid System Falsification}

\author{
    Gidon Ernst
    \institute{LMU Munich, Germany}
    \email{gidon.ernst@lmu.de}
\and
    Ichiro Hasuo
    \institute{National Institute of Informatics, Tokyo, Japan} 
    \email{hasuo@nii.ac.jp}
\and
    Zhenya Zhang
    \institute{Kyushu University, Japan} 
    \email{zhang.zhenya.623@m.kyushu-u.ac.jp}
\and Sean Sedwards
    \institute{University of Waterloo, Waterloo, Canada}
    \email{sean.sedwards@uwaterloo.ca }
}

\def\authorrunning{G. Ernst, I. Hasuo, Z. Zhang \& S. Sedwards}

\begin{document}

\maketitle

\begin{abstract}
\emph{Optimization-based falsification}  employs stochastic optimization algorithms to search for error input of hybrid systems.
In this paper we introduce a simple idea to enhance falsification, namely \emph{time staging}, that allows the \emph{time-causal} structure of time-dependent signals to be exploited by the optimizers.
Time staging consists of running a falsification solver multiple times, from one interval to another, incrementally constructing an input signal candidate.
Our experiments show that time staging  can dramatically increase performance in some realistic examples.
We also present theoretical results that suggest the kinds of models and specifications for which time staging is likely to be effective.
\end{abstract}

\section{Introduction}
\label{sec:introduction}

\paragraph{Hybrid Systems}
Quality assurance of \emph{cyber-physical systems} (CPS) has been recognized as an important challenge,
where many CPS are {\em hybrid systems} that
combine the discrete dynamics of computers and the continuous dynamics of physical components. 
Unfortunately, analysis of hybrid systems poses unique challenges, such as the limited applicability of \emph{formal verification}.
In formal verification one aims to give a mathematical proof for a system's correctness. This is much harder for hybrid systems than for computer software/hardware,
 where the presence of continuous dynamics makes many problems more complex or even undecidable (e.g.\ reachability in hybrid automata).

\vspace*{.3em}
   \begin{wrapfigure}[7]{r}{0pt}
     \includegraphics[width=.6\textwidth]{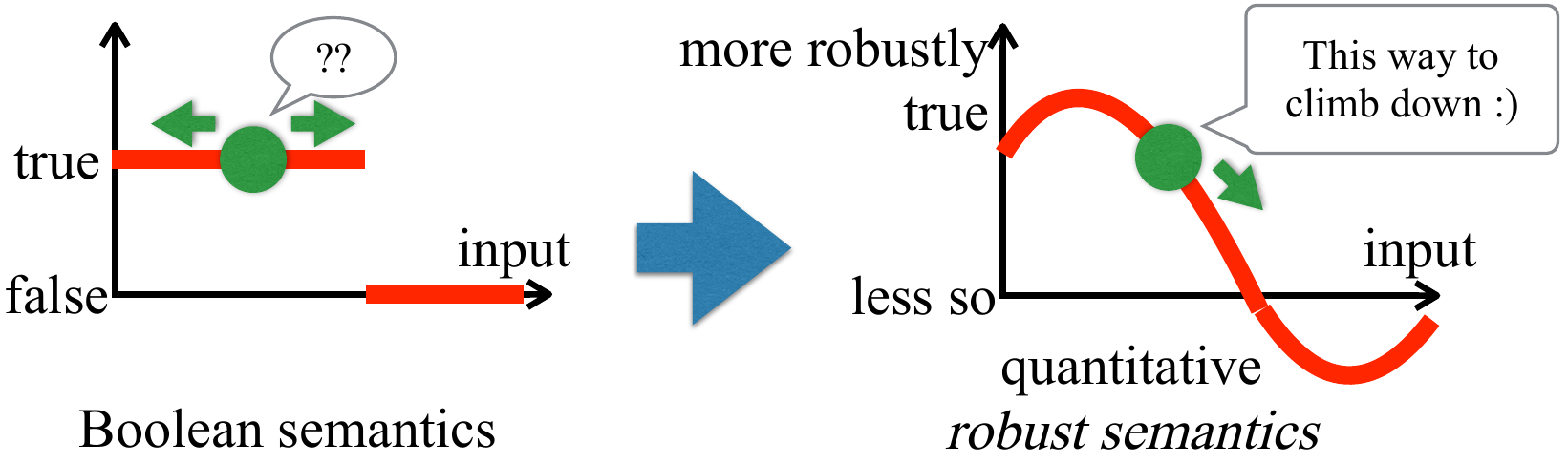}
    \caption{From Boolean to robust semantics}
    \label{fig:fromQualitativeToQuantitative}
\end{wrapfigure}
\paragraph{Optimization-Based Falsification}
Because of these difficulties, an increasing number of  researchers are turning to  \emph{opti\-mization-based falsification} as a quality assurance measure. It is a testing method rather than that of formal verification; the problem is formalized as follows. 
\begin{center}
   \begin{minipage}{0.65\textwidth}
    \begin{itemize}
    \item{\textbf{Given:}} 
      a \emph{model} $\mathcal{M}$ (that takes an input signal $\bu$
      and  yields  an output signal $\mathcal{M}(\bu)$), and
      a \emph{specification} $\varphi$ (a temporal formula)
    \item{\textbf{Answer:}} 
      \emph{error input}, that is, an input signal $\bu$ such
      that the corresponding output $\mathcal{M}(\bu)$ violates $\varphi$ 
    \end{itemize}
  \end{minipage}
  \begin{math}
   \xymatrix@1@+1.5em{
   {}
     \ar[r]^-{\bu}
   &
   {\xybox{ *++[F]{\mathcal{M}} }}
     \ar[r]^-{\mathcal{M}(\bu)}_-{\not\models\varphi \; ?}
   &
   {}
   }
  \end{math}
\end{center}
In the optimization-based falsification approach, the above falsification problem is turned into an optimization problem. This is possible thanks to \emph{robust semantics} of temporal formulas~\cite{DBLP:journals/tcs/FainekosP09}. Instead of the Boolean satisfaction relation $\bv\models\varphi$, robust semantics assigns a quantity $\sem{\bv,\varphi}\in\R\cup\{\infty,-\infty\}$ that tells us, not only whether $\varphi$ is true or not (by the sign), but also \emph{how robustly} the formula is true or false. This allows one to employ hill-climbing optimization (see Fig.~\ref{fig:fromQualitativeToQuantitative}): we iteratively generate input signals, in the direction of decreasing robustness, hoping that eventually we hit negative robustness. 

 Optimization-based falsification is a subclass of \emph{search-based testing}: it adaptively chooses test cases (input signals $\bu$) based on previous observations. One can use stochastic algorithms for optimization (such as simulated annealing), which turn out to be much more scalable than many model checking algorithms that rely on exhaustive search. Note also that the system model $\mathcal{M}$ can be  black box: it is enough to know the correspondence between input $\bu$ and output $\mathcal{M}(\bu)$.
 An error input $\bu$ is concrete evidence for the system's need for improvement, and thus has great appeal to practitioners.

The approach of  optimization-based falsification was initiated in~\cite{DBLP:journals/tcs/FainekosP09} and has been actively pursued ever since~\cite{DBLP:conf/tacas/AnnpureddyLFS11,DBLP:conf/cav/AdimoolamDDKJ17,DBLP:conf/atva/DeshmukhJKM15,DBLP:conf/formats/KuratkoR14,DBLP:conf/cav/Donze10,DBLP:conf/formats/DonzeM10,DBLP:conf/nfm/DreossiDDKJD15}. There are now mature tools, such as Breach~\cite{DBLP:conf/cav/Donze10} and S-Taliro~\cite{DBLP:conf/tacas/AnnpureddyLFS11},
which work with industry-standard Simulink models.

\paragraph{Contribution}
 We introduce
a simple idea of \emph{time staging} for  enhancement of optimization-based falsification.
Time staging consists of running a falsification solver repeatedly, from one input segment to another,  incrementally constructing an  input signal. 

In general, in solving a concrete problem $\mathcal{C}$ by a metaheuristic $\mathcal{H}$ (such as stochastic optimization, evolutionary computation, etc.), 
a key to success
is to communicate as much information as possible in the translation from $\mathcal{C}$ to $\mathcal{H}$---that is, to let $\mathcal{H}$ exploit  structures unique to $\mathcal{C}$. 
Our idea of time staging follows this philosophy. More specifically, via time staging we communicate the \emph{time-causal} structure of time-dependent signals---a structure that is present in some instances of the falsification problem but not in optimization problems in general---to stochastic optimization solvers.

Our implementation of time-staged falsification is
based on Breach~\cite{DBLP:conf/cav/Donze10}. We show that this simple idea can dramatically enhance its performance in some examples.
We also present some theoretical considerations on the kinds of problem instances where time staging is likely to work, and some results that aid implementation of time staging.

\paragraph{Structure of the Paper} In~\S{}\ref{sec:preview} we informally outline optimization-based falsification and illustrate the idea of time staging. We then turn to formal developments:
in~\S{}\ref{sec:unstagedFalsification} we review existing falsification works and in~\S{}\ref{sec:timeStagedFalsification} we present our algorithm, augmented by some theoretical results that aid its implementation.
\S{}\ref{sec:theoreticalBoundaryCaseResults} is devoted to the theoretical consideration of two specific settings in which time staging is guaranteed to work well. These settings will serve as useful ``rules of thumb'' for practical applications.
In~\S{}\ref{sec:experiments} we discuss our implementation and experimental results.
We conclude in~\S{}\ref{sec:conclusion}.

\section{Schematic Overview:  Falsification and Time Staging}
\label{sec:preview}
\begin{figure}[tbp]
\centering
  \includegraphics[scale=.6]{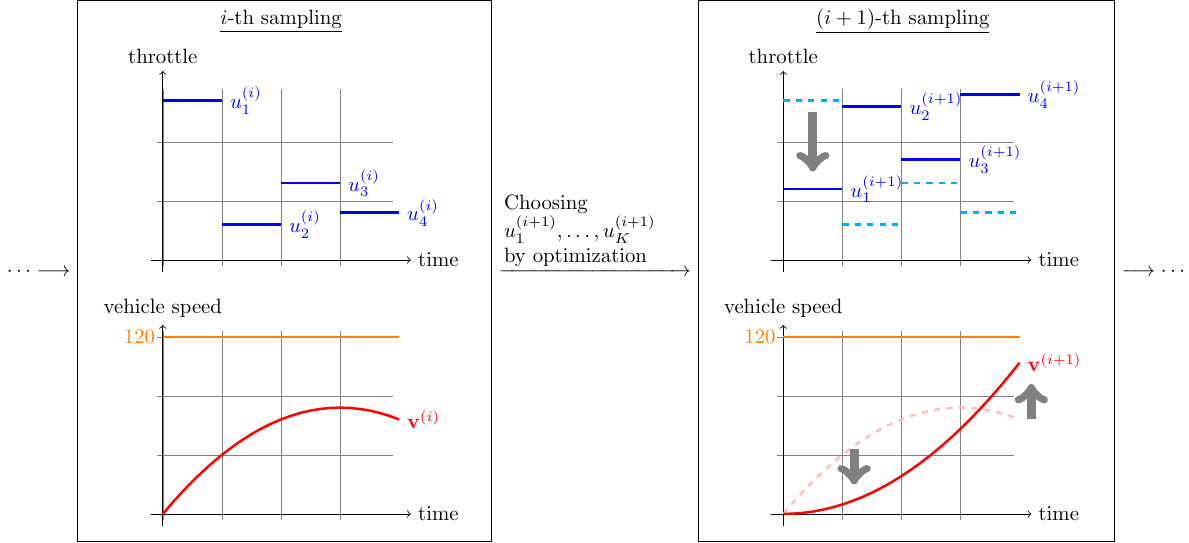}
\caption{Conventional optimization-based falsification (without time staging). }
\label{fig:nonstaged}

\vspace*{1.5em}
\centering
 \includegraphics[scale=.6]{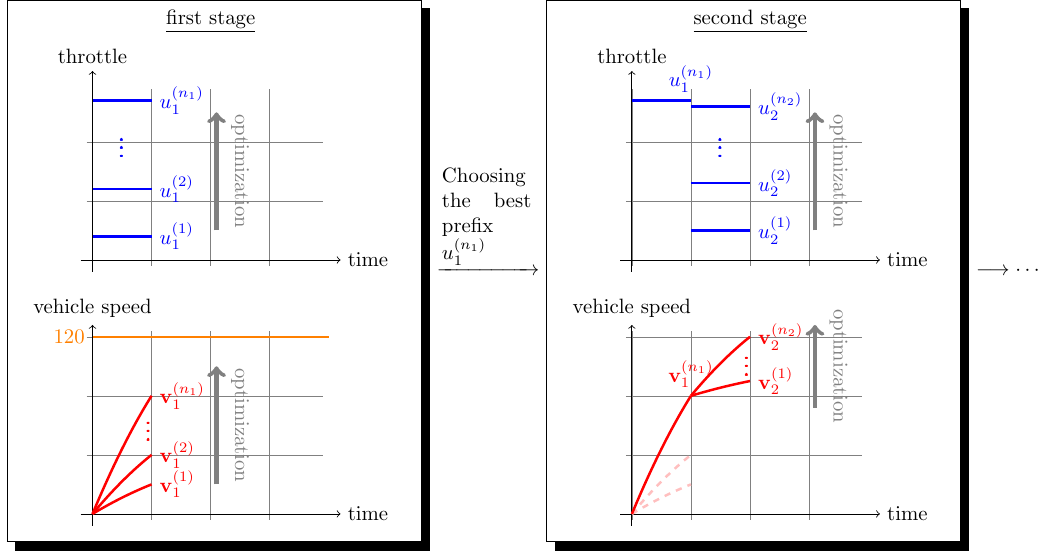}
\caption{Falsification with time staging}
\label{fig:staged}

\vspace*{1.5em}
\centering
\begin{minipage}{.5\textwidth}
  \includegraphics[scale=.73]{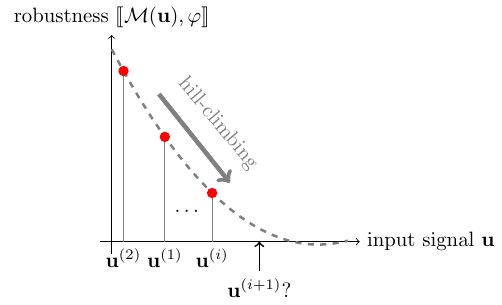}
 \caption{Hill-climbing optimization in falsification}
 \label{fig:hillClimbing}
\end{minipage}
\;
\begin{minipage}{.48\textwidth}
 \begin{tabular}{l}
 \raisebox{-.5\height}{\includegraphics[width=.3\textwidth]{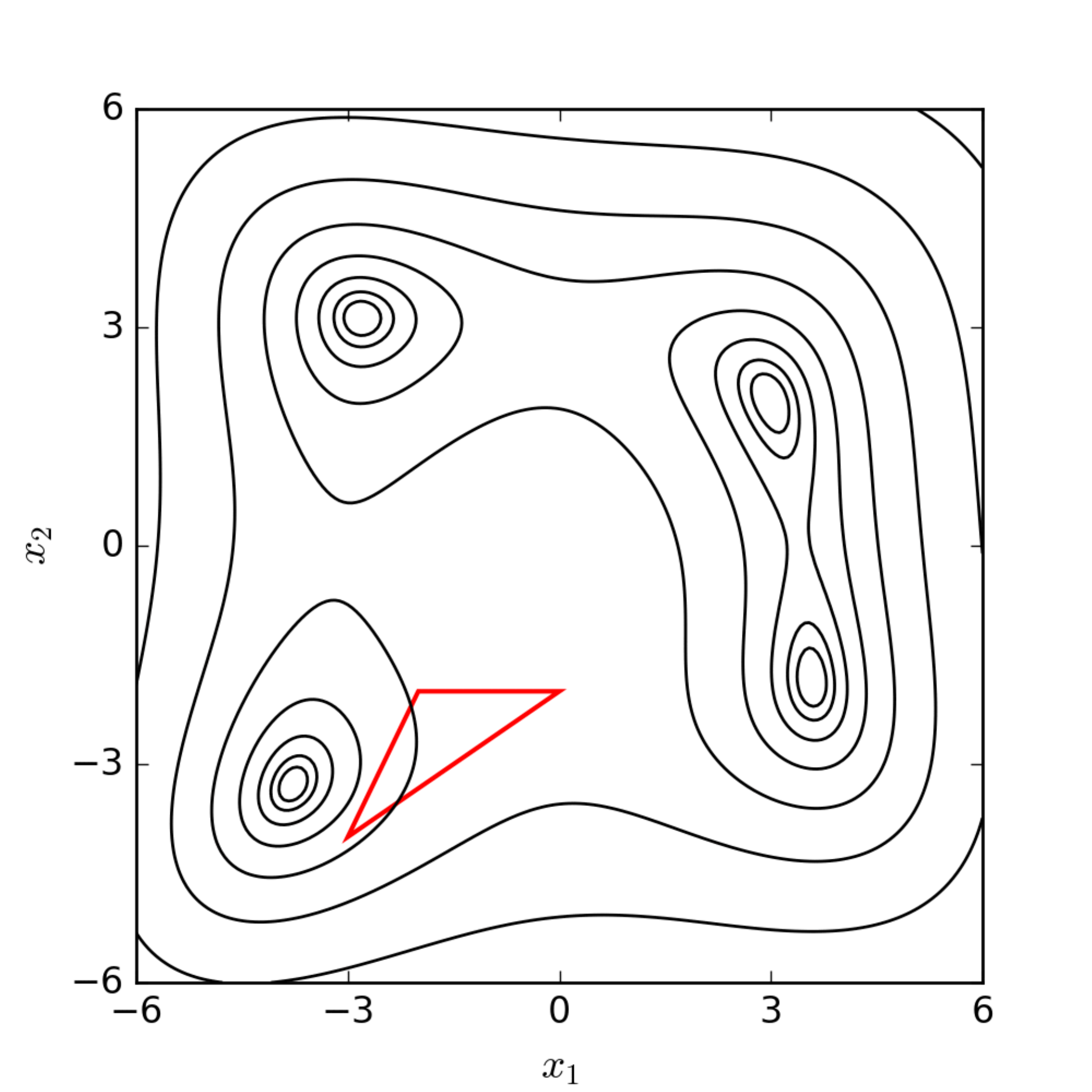}}
 $\Rightarrow$
 \raisebox{-.5\height}{\includegraphics[width=.3\textwidth]{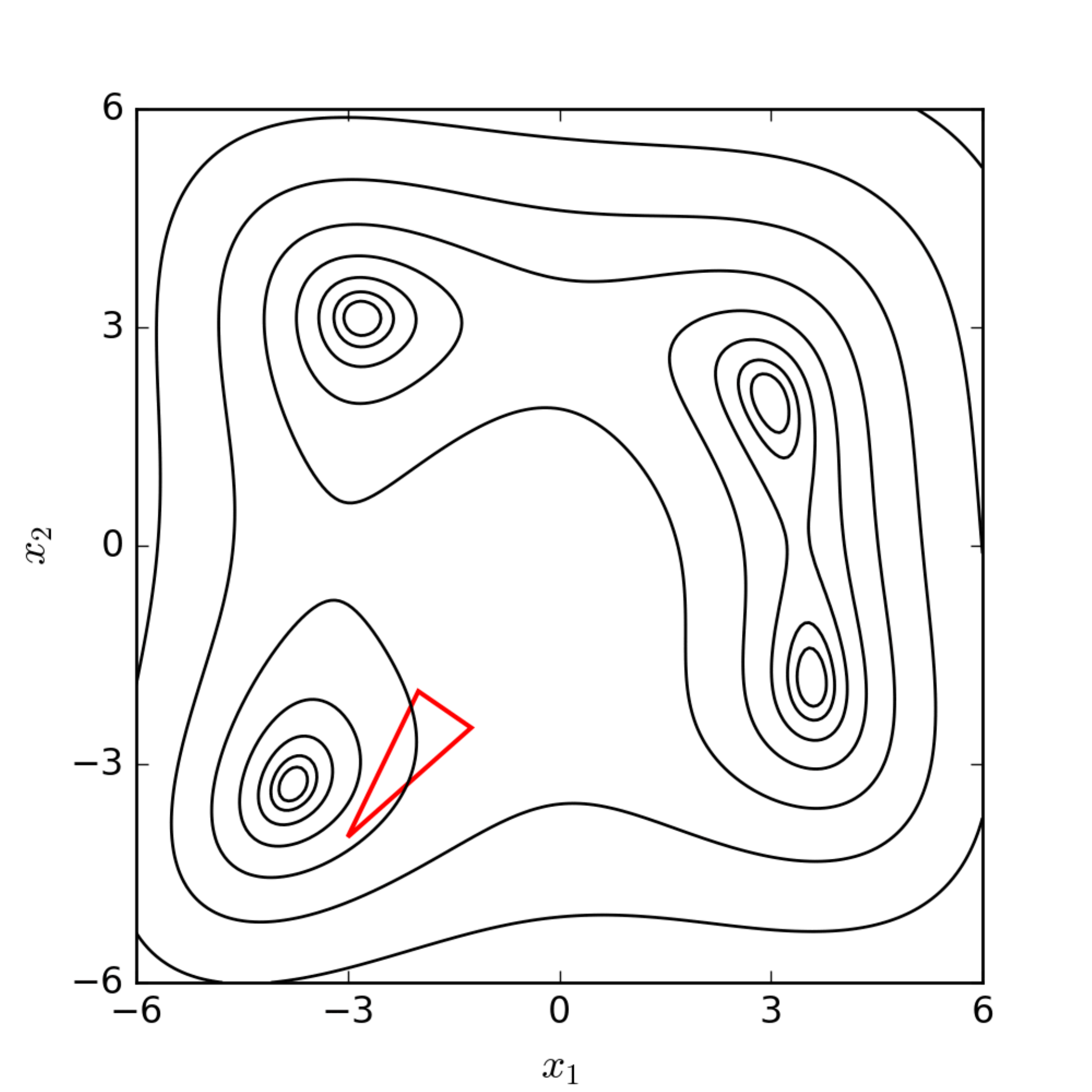}}
 $\Rightarrow$
 \\
 \raisebox{-.5\height}{\includegraphics[width=.3\textwidth]{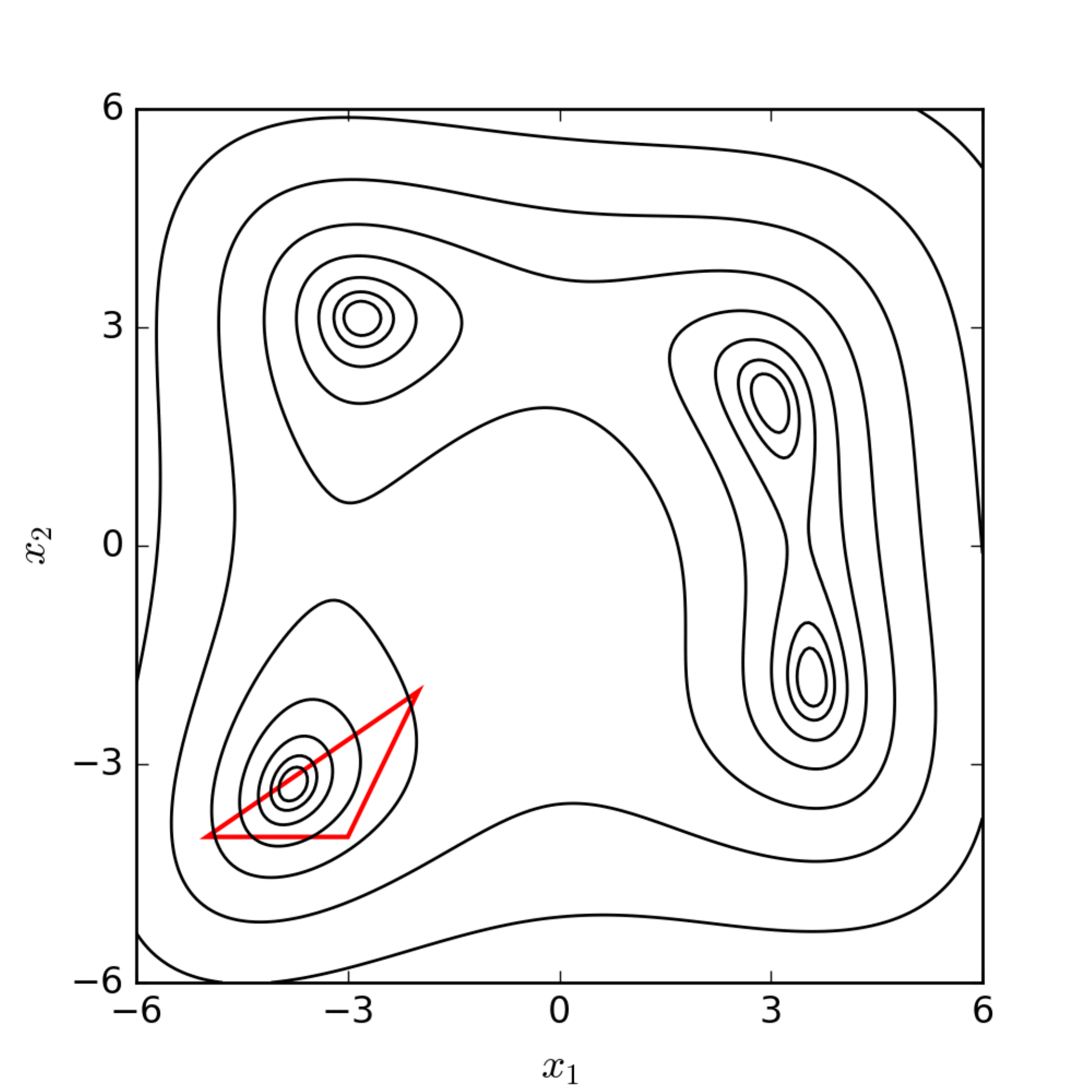}}
 $\Rightarrow\cdots\Rightarrow$ 
 \raisebox{-.5\height}{\includegraphics[width=.3\textwidth]{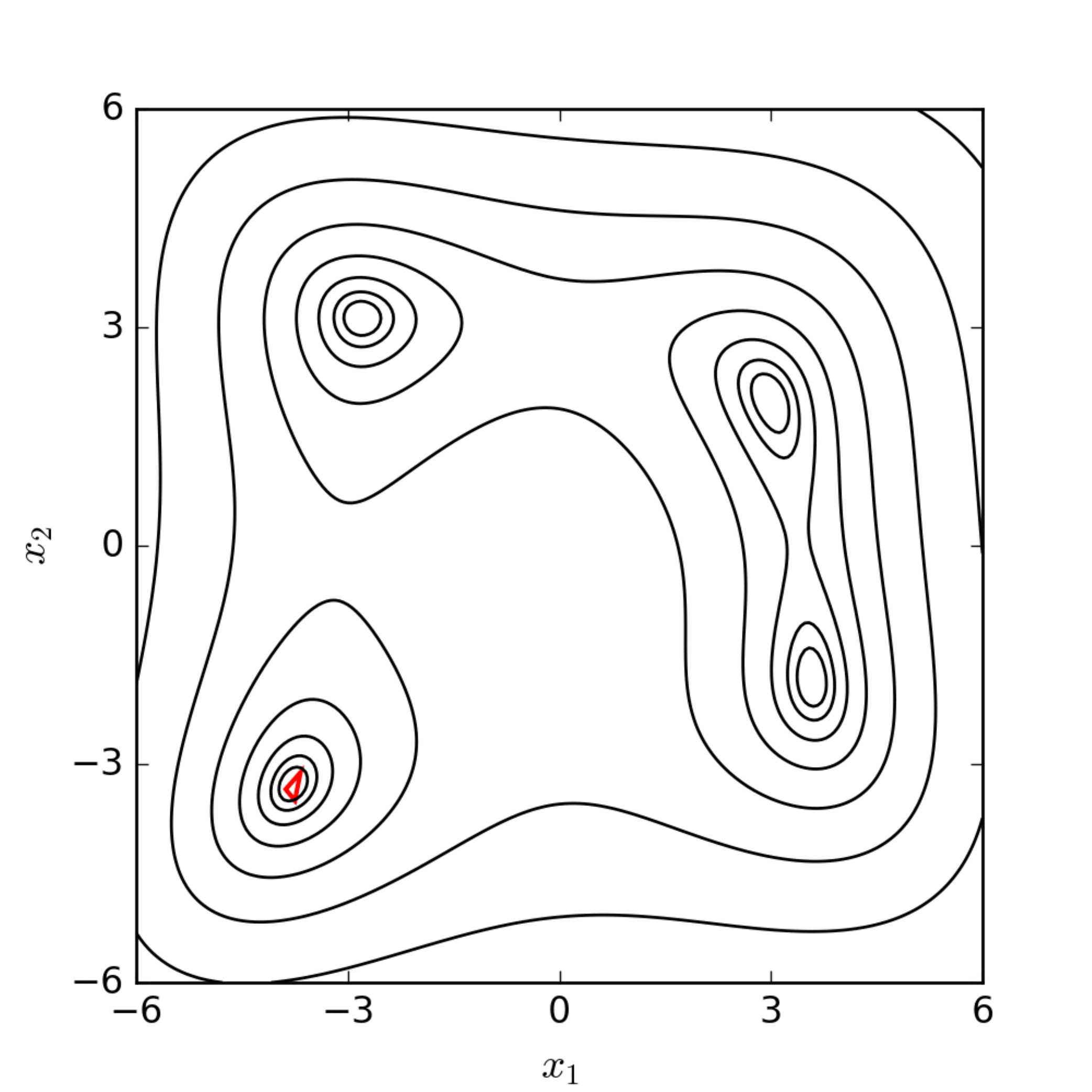}}
 \end{tabular}
 \caption{Nelder-Mead optimization. 
Here the input space is the two-dimensional square, and the (unknown) score function is depicted by contour lines. 
Figures are from Wikipedia
}
 \label{fig:neldermead}
\end{minipage}
\end{figure}
 We illustrate  falsification and time staging, informally with an example.

\paragraph{Example Setting}
 We take as a system model $\mathcal{M}$ a simple automotive powertrain whose input
signal is $1$-dimensional (the throttle $u$) and whose output
signal is the vehicle speed $v$. We assume that $\mathcal{M}$ exhibits
the following natural behavior: the larger $u$ is, the quicker $v$
grows. Let our specification be $\varphi\equiv\bigl(\Box(v\le
120)\bigr)$, where $\equiv$ denotes the syntactic equality.  To falsify $\varphi$ the vehicle speed $v$
must exceed $120$.
 From the assumption about $\mathcal{M}$'s behavior,
we expect  $u$ to be large in a falsifying input signal.
Note that this is a simplified version of one of our experiments in~\S{}\ref{sec:experiments}.

\paragraph{Optimization-Based Falsification}
Fig.~\ref{fig:nonstaged} illustrates how a conventional optimization-based
falsification procedure
works. In the $i$-th
sampling one tries an input signal $\bu^{(i)}$.
Following the falsification literature we focus on piecewise constant signals. Thus a signal
$\bu^{(i)}$ is represented by a sequence
$(u_{1}^{(i)},\dotsc,u_{K}^{(i)})$ of real numbers. See the top left of Fig.~\ref{fig:nonstaged}.
The corresponding output signal $\bv^{(i)}=\mathcal{M}(\bu^{(i)})$ is shown below it.

Since $\bv^{(i)}$ does not reach the threshold $120$, we move on to the $(i+1)$-th sampling and try a new input signal 
$\bu^{(i+1)}=(u_{1}^{(i+1)},\dotsc,u_{K}^{(i+1)})$.
The choice of $\bu^{(i+1)}$ is made by an optimization algorithm.
Specifically, the optimization algorithm observes the results of the previously sampled  input signals $\bu^{(1)},\dotsc,\bu^{(i)}$---especially the
robustness value $\sem{\mathcal{M}(\bu^{(i)}),\varphi}$ that each input
$\bu^{(i)}$ achieves. 
In the current setting where $\varphi\equiv\bigl(\Box(v\le 120)\bigr)$,
 the robustness value is simply the difference between $120$ and the peak vehicle speed.
 The optimization algorithm tries to derive some general tendency, which it then uses
 to increase the probability that the next input signal $\bu^{(i+1)}$ will make the robustness smaller (i.e.\ the peak vehicle speed higher). 

\emph{Hill climbing}  is a prototype of such optimization algorithms.
Its use in falsification is illustrated in Fig.~\ref{fig:hillClimbing}, where $\bu=(u_{1},\dotsc,u_{K})$ is depicted as one-dimensional for clarity.
The actual curve for the robustness value $\sem{\mathcal{M}(\bu),\varphi}$ (gray and dashed) is unknown. Still the previous observations under input $\bu^{(1)},\dotsc,\bu^{(i)}$ suggest that  to the right is the  climbing down direction.  The next candidate $\bu^{(i+1)}$ is picked accordingly, towards negative robustness. Another well-known optimization algorithm is the \emph{Nelder-Mead algorithm}. See Fig.~\ref{fig:neldermead}, where the input space is two-dimensional and the (unknown) robustness function is depicted by contour lines.

We see in the right of Fig.~\ref{fig:nonstaged}
that the new input signal $\bu^{(i+1)}=(u^{(i+1)}_{1},\dotsc,u^{(i+1)}_{K})$
leads to a corresponding output signal $\bv^{(i+1)}$
that reduces the robustness value by achieving a higher peak speed. 
We continue this way,
 $\bu^{(i+2)},\bu^{(i+3)},\dotsc$, hoping to eventually reach a falsifying input signal. 

\paragraph{Absence of the Time-Causal Information}
A closer look at 
Fig.~\ref{fig:nonstaged}  reveals room for
improvement. In Fig.~\ref{fig:nonstaged}, the new input signal
\raisebox{0pt}[0pt]{$\bu^{(i+1)}$} indeed achieves a smaller overall robustness 
$\sem{\mathcal{M}(\bu),\varphi}$ than \raisebox{0pt}[0pt]{${\bu^{(i)}}$}. 
 However, its initial segment \raisebox{0pt}[0pt]{${u^{(i+1)}_{1}}$} is  smaller than \raisebox{0pt}[0pt]{${u^{(i)}_{1}}$};
 consequently the vehicle speed \raisebox{0pt}[0pt]{${\bv^{(i+1)}}$} is smaller than \raisebox{0pt}[0pt]{${\bv^{(i)}}$} in the first few seconds.
 Keeping the initial segment \raisebox{0pt}[0pt]{${u^{(i)}_{1}}$} would have achieved an even greater peak speed.

The problem here is that the \emph{time-causal} structure inherent in the problem is not explicitly communicated to the optimization algorithm. The relevant structure is more specifically  \emph{time monotonicity}:  an input prefix that achieves smaller robustness (i.e.\ a greater peak speed) is more likely to extend to a full falsifying input signal. Although it is possible that a stochastic optimization algorithm somehow ``learns'' time monotonicity, it is not guaranteed, because the structure of input spaces (the horizontal axis in Fig.~\ref{fig:hillClimbing} and the squares in Fig.~\ref{fig:neldermead})
 does not explicitly reflect time-causal structures.

While the time monotonicity is not shared by all instances of the falsification problems, we find many realistic instances that approximately satisfy the property.
We discuss time monotonicity in~\S{}\ref{sec:theoreticalBoundaryCaseResults}, as well as in the context of our experiments in~\S{}\ref{sec:experiments}.

\paragraph{Falsification with Time Staging}
Our proposal of \emph{time staging} consists of incrementally synthesizing a candidate input signal. 
We illustrate this in Fig.~\ref{fig:staged}. 
In the first stage (left), we run a falsification algorithm and try to find an initial input segment that achieves low robustness (i.e.\ high peak speed). This first stage comprises running $n_{1}$ samplings, as illustrated in Fig.~\ref{fig:nonstaged}.
 This process will gradually improve candidates for the initial input segment, in the way the  arrows $\uparrow$ on the left in Fig.~\ref{fig:staged} 
  designate. Let us assume that the last candidate \raisebox{0pt}[0pt]{$u_{1}^{(n_{1})}$} is the (tentative) best, achieving the smallest robustness.

In the second stage (on the right in Fig.~\ref{fig:staged}) we continue  \raisebox{0pt}[0pt]{$u_{1}^{(n_{1})}$} and synthesize the second input segment. This is again by running  a falsification algorithm, as depicted. Note that, in each stage (a box in Fig.~\ref{fig:staged}), the whole iterated process in Fig.~\ref{fig:nonstaged} is conducted.
In this way we continue to the $K$-th stage, always starting with the input segment that performed the best in the previous stage, thus exploiting the time-causal structure.

While time staging is not difficult to implement, there is a challenge in using it effectively.
An immediate question is whether choosing the single best input segment in each stage is the optimal approach.
Our current strategy favors exploitation over exploration: it might miss a falsifying signal whose robustness must decrease slowly in the earlier segments and only quickly in the latter segments. Indeed we are working on an evolutionary variant of the above time-staged algorithm, where multiple segments are passed over from one stage to another, in order to maintain diversity and conduct exploration. That said, even under the current simple strategy of picking the best one, we observe significant performance enhancement in some falsification problems. See~\S{}\ref{sec:experiments}. 

We can summarize this trade-off in terms of the size of search spaces.
Let $U$ be the set of  candidates for input segments, and $K$ be the number of stages. Then the size of the set of whole input signals is $|U|^{K}$, choosing one input segment for each stage. In our staged algorithm, in contrast, the search space for each stage is $U$ and overall our search space is $K\cdot |U|$. This reduction comes with the risk of missing some falsifying input signals. The experimental results in~\S{}\ref{sec:experiments} suggest this risk is worth taking.
Moreover, in~\S{}\ref{sec:theoreticalBoundaryCaseResults} we present some theoretical conditions for the absence of such risk.
They help users decide in practical applications when time staging will be effective.

\section{Optimization-Based Falsification }\label{sec:unstagedFalsification}
From this section on we turn to the formal description and analysis of our algorithm.  This section presents a review of  existing works on optimization-based falsification.

\paragraph{System Models}
Let us formalize our system models. 

 \begin{mydefinition}[time-bounded signal]\label{def:operationsOnSignals}
 Let $T\in \Rpos$ be a positive real. 
 A \emph{(time-bounded) $m$-dimensional signal} with a time horizon $T$ is a function $\bw\colon [0,T]\to\Rm$. 

 Let $\bw\colon [0,T]\to \Rm$ and $\bw'\colon [0,T']\to\Rm$ be (time-bounded) signals. Their \emph{concatenation} $\bw\cdot\bw'\colon [0,T+T']\to \Rm$ is defined by
 \begin{math}
  (\bw\cdot\bw')(t):=
 \bw(t)
 \end{math}
 if $t\in [0,T]$, and
   $\bw'(t-T)$ if $t\in(T,T+T']$.

 Let $T_{1},T_{2}\in (0,T]$ such that $T_{1}<T_{2}$. The \emph{restriction} 
 $\bw|_{[T_{1},T_{2}]}\colon [0,T_{2}-T_{1}]\to \Rm$ of  $\bw\colon [0,T]\to \Rm$ to the interval  $[T_{1},T_{2}]$ is defined by $(\bw|_{[T_{1},T_{2}]})(t):=\bw(T_{1}+t)$. 
 \end{mydefinition}

\begin{mydefinition}[system model $\mathcal{M}$]\label{def:systemModel}
 A \emph{system model}, with $M$-dimensional input,  is a function $\mathcal{M}$ that takes 
an input signal $\bu\colon [0,T]\to \R^{M}$ and returns $\mathcal{M}(\bu)\colon [0,T]\to \R^{N}$. Here the common time horizon $T\in \Rpos$ 
is arbitrary. 

Furthermore, we impose the following \emph{causality} condition on $\mathcal{M}$.
For any time-bounded signals $\bu\colon [0,T]\to \R^{M}$ and $\bu'$, we require that
\begin{math}\label{eq:causality}
 \mathcal{M}(\bu\cdot\bu')
\big|_{[0,T]}
= \mathcal{M}(\bu)
\end{math}.
\end{mydefinition}
Note that $
\mathcal{M}(\bu\cdot\bu')
=
\mathcal{M}(\bu)
\cdot
\mathcal{M}(\bu')
$
does not hold in general: feeding $\bu$ can change the internal state of $\mathcal{M}$. This motivates the following definition. 

\begin{mydefinition}[continuation $\mathcal{M}_{\bu}$]\label{def:continuation}
Let $\mathcal{M}$ be a system model and $\bu\colon [0,T]\to \R^{M}$ be a signal. The \emph{continuation} of $\mathcal{M}$ after $\bu$, denoted by $\mathcal{M}_{\bu}$, is defined as follows. For an input signal $\bu'\colon [0,T']\to \R^{M}$:
\begin{math}
 \mathcal{M}_{\bu}(\bu')(t)
 :=
 \mathcal{M}(\bu\cdot\bu')(T+t)
\end{math}.
\end{mydefinition}
\auxproof{ The notions and conditions in Def.~\ref{def:systemModel}--\ref{def:continuation} allow simpler presentation as \emph{Mealy machines}, in case the notion of time is discrete. See Appendix~\ref{appendix:mealymachine}. 
}

\paragraph{Signal Temporal Logic and Robust Semantics}
We review \emph{signal temporal logic} (\STL)~\cite{DBLP:conf/formats/MalerN04} and its \emph{robust semantics}~\cite{DBLP:journals/tcs/FainekosP09,DBLP:conf/formats/DonzeM10}.
 $\Var$ is the set of variables, and let $N:=|\Var|$. 
 Variables stand for physical quantities,
control modes, etc.  $\equiv$ denotes syntactic equality. 

\begin{mydefinition}[syntax]\label{def:syntax}
  In $\STL$, 
 \emph{atomic propositions} and 
 \emph{formulas} are defined as follows, respectively:
\begin{math}
       \alpha 
 \,::\equiv\,
        f(x_1, \dots, x_n) > 0
\end{math}, and 
\begin{math}
      \varphi  \,::\equiv\,
        \alpha \mid \bot
        \mid \neg \varphi 
        \mid \varphi \wedge \varphi 
        \mid \varphi \UntilOp{I} \varphi
\end{math}. Here
 $f$ is an $n$-ary function $f:\R^n \to \R$, $x_1, \dots, x_n \in \Var$,
  and $I$ is a closed non-singular interval in $\Rnn$,
  i.e.\ $I=[a,b]$ or $[a, \infty)$ where $a,b \in \R$ and  $a<b$.
\end{mydefinition}

\auxproof{$\STL$ is  much like LTL: it lacks the next-time operator (because of continuous time); and it allows temporal modalities restricted to intervals $I$.  
}

  We omit subscripts $I$ for temporal operators if $I = [0, \infty)$. Other common connectives like $\lor,\rightarrow,\top$, $\Box_{I}$ (always) and $\Diamond_{I}$ (eventually), are introduced as abbreviations: $\Diamond_{I}\varphi\equiv\top\UntilOp{I}\varphi$ and 
$\Box_{I}\varphi\equiv\lnot\Diamond_{I}\lnot\varphi$.  Atomic formulas like $f(\vec{x})\le c$, where $c\in\R$ is a constant, are also accommodated by using negation and the function $f'(\vec{x}):=f(\vec{x})-c$. 
\begin{mydefinition}[robust semantics~\cite{DBLP:conf/formats/DonzeM10,DBLP:conf/cav/DonzeFM13}]\label{def:semantics}
  For an unbounded $n$-dimensional signal $\bw \colon \Rnn\to \R^{n}$ and $t\in \Rnn$,
 $\bw^t$ denotes the \emph{$t$-shift} of $\bw$, that is, 
 $\bw^t(t') \Defeq \bw(t+t')$.

  Let $\bw \colon \Rnn \to \R^{N}$ be a signal (recall $N=|\Var|$),
  and $\varphi$ be an $\STL$ formula.
  We define the \emph{robustness} 
  $\Robust{\bw}{\varphi} \in \R \cup \{\infty,-\infty\}$ 
as follows, by induction.
  Here $\sqcap$ and $\sqcup$ denote infimums and supremums of real numbers, respectively.
    \begin{displaymath}
      \begin{array}{l}
          \Robust{\bw}{f(x_1, \cdots, x_n) > 0}  \;\Defeq \;
           f\bigl(\bw(0)(x_1), \cdots, \bw(0)(x_n)\bigr) 
\qquad
          \Robust{\bw}{\bot}  \;\Defeq\;  -\infty
\\
          \Robust{\bw}{\neg \varphi}   \;\Defeq\;   - \Robust{\bw}{\varphi}\qquad
          \Robust{\bw}{\varphi_1 \wedge \varphi_2}   \;\Defeq\;   \Robust{\bw}{\varphi_1} \sqcap \Robust{\bw}{\varphi_2}\\
          \Robust{\bw}{\varphi_1 \UntilOp{I} \varphi_2}   \;\Defeq\; 
                                                         \textstyle{ \Vee{t \in I}\bigl(\,\Robust{\bw^t}{\varphi_2} \sqcap 
                                                         \Wedge{t' \in [0, t)} \Robust{\bw^{t'}}{\varphi_1}\,\bigr)}
      \end{array}
    \end{displaymath}
\end{mydefinition}

Here are some intuitions and consequences. 
The robustness $\Robust{\bw}{f(\vec{x})>c}$ stands for the vertical margin $f(\vec{x})-c$ for the signal $\bw$ at time $0$.  A negative robustness value indicates how far the formula is from being true.
The robustness for the eventually modality is computed by
\begin{math}
       \Robust{\bw}{\DiaOp{[a,b]} (x > 0)}
      = \Vee{t \in [a,b]} \bw(t)(x)
\end{math}.

The original semantics of $\STL$ is Boolean, given by a binary relation $\models$ between signals and formulas. The robust semantics refines the Boolean one, in the sense that:
$ \sem{\bw,\varphi} > 0$
 implies
$\bw\models\varphi$, and
$ \sem{\bw,\varphi} < 0$
implies
$\bw\not\models\varphi$. 
  Optimization-based falsification via robust semantics~\cite{DBLP:journals/tcs/FainekosP09} hinges on this refinement. See~\cite{DBLP:conf/formats/DonzeM10}. 

Although the definitions so far are for unbounded signals only,
we note that the robust semantics $\sem{\bw,\varphi}$,
as well as the Boolean satisfaction $\bw\models\varphi$, allows straightforward adaptation to time-bounded signals (Def.~\ref{def:operationsOnSignals}). See Appendix~\ref{appendix:semanticsForTimeBoundedSignals}. 

\auxproof{Take a specification $\DiaOp{[0,5](x>0)}$ as an example; unlike Def.~\ref{def:semantics}, time robustness is sensitive to how quickly the requirement $x>0$ gets satisfied. Our idea of time staging can be applied to all of these different robustness notions; in this paper we focus on the simplest setting of space robustness for the sake of presentation.}

\paragraph{Falsification Solvers}
In the next definition, 
a prototype of a score function~$\rho$ is given by   the robustness $\rho_{\varphi}$
of a given $\STL$ specification
$\varphi$. The generality of allowing other
$\rho$ is needed later in~\S{}\ref{sec:timeStagedFalsification}. 
\begin{equation}\label{eq:scoreByRobustness}
 \rho_{\varphi}(\bv)\;:=\;\sem{\bv,\varphi}
\end{equation} 

\begin{mydefinition}[falsification solver]
\label{def:falsificationSolver}
A \emph{falsification solver} is a stochastic algorithm 
\begin{math}
 \Falsify
\end{math}
that takes, as input: 
1) a system model $\mathcal{M}$ (Def.~\ref{def:systemModel}) with $M$-dimensional input;
2) a score function~$\rho$ that takes an output signal $\bv$ of
        $\mathcal{M}$ and returns a score
        $\rho(\bv)\in\R\cup\{-\infty,\infty\}$; and
3)
a time horizon $T\in\Rpos$.
The algorithm $\Falsify$ returns an $M$-dimensional signal 
$\bu\colon [0,T]\to \R^{M}$. 

 Each invocation $\Falsify(\mathcal{M},
 \rho,
 T)$ of the solver is called a \emph{falsification trial}. It is \emph{successful} if the returned signal $\bu$ satisfies 
\begin{math}
  \rho(\mathcal{M}(\bu)) < 0
 \end{math}. 
 Note that the returned signal~$\bu$ can differ in every trial,
since $\Falsify$ is a stochastic algorithm. 

We further assume the internal structure of the solver $\Falsify$
 follows the scheme in Algorithm~\ref{algorithm:falsify}. It consists
 of two phases. The first \emph{initial sampling} phase collects
 some  candidates for $\bu\colon [0,T]\to\R^{M}$ regardless of the
 system model $\mathcal{M}$ or the score function~$\rho$. In the second
 \emph{optimization sampling} phase, a stochastic optimization algorithm
 is employed to sample a candidate $\bu$ that is likely to make
 the score
 $\rho(\mathcal{M}(\bu))$
 small.
\begin{algorithm}[t]
\caption{Internal Structure of a Falsification Solver $\Falsify(\mathcal{M},\rho,T)$}
\label{algorithm:falsify}
\begin{algorithmic}[1]
 \Require a system model $\mathcal{M}$,
 a score function $\rho$,
 and $T\in\Rpos$
\State $\mathsf{U}\gets ()$  
   \Comment{the list $\mathsf{U}$ collects all the candidates $\bu\colon [0,T]\to \R^{M}$}
 \While {$\lnot\mathsf{InitialSamplingDone}(T,\mathsf{U})$}
 \label{line:initSamplingStart}
 \State $\bu\gets\mathsf{InitialSampling}(T)$
       \Comment{ $\bu\colon [0,T]\to \R^{M}$ is sampled following some recipe}         
 \State $\mathsf{U}\gets \mathsf{cons}(\mathsf{U},\bu)$
  \label{line:initSamplingEnd}
\EndWhile
\While {$\lnot\mathsf{OptimizationSamplingDone}(\mathcal{M},\rho,T,\mathsf{U})$}
  \State $\bu\gets\mathsf{OptimizationSampling}(\mathcal{M},\rho,T,\mathsf{U})$
    \State 
     \Comment{\parbox[t]{.92\linewidth}{
           $\bu$ is sampled, so that 
           $\rho(\mathcal{M}(\bu))$
           becomes small,
           based on  previous samples in $\mathsf{U}$}}
 \State $\mathsf{U}\gets \mathsf{cons}(\mathsf{U},\bu)$
\EndWhile
\State $\bu\gets\argmin_{\bu\in\mathsf{U}}
   \rho(\mathcal{M}(\bu))
$
\State \Return $\bu$
 \Comment{a trial is successful if 
 $
 \rho(\mathcal{M}(\bu))
 < 0$}
\end{algorithmic}
\end{algorithm}
\end{mydefinition}

\paragraph{Implementation of Falsification Solvers}
Both Breach~\cite{DBLP:conf/cav/Donze10}  and S-Taliro~\cite{DBLP:conf/tacas/AnnpureddyLFS11} take industry-standard Simulink models as system models. For input signal candidates the tools focus on piecewise constant signals; they are represented by sequences $(u_{1},\dotsc,u_{K})$ of real numbers, much like in~\S{}\ref{sec:preview}. Here $K$ is the number of \emph{control points}; in our staged algorithm we use the same $K$ for the number of stages. 

The tools offer multiple  stochastic optimization algorithms for  the optimization sampling phase, including \emph{CMA-ES}~\cite{DBLP:conf/cec/AugerH05},  \emph{global Nelder-Mead} and \emph{simulated annealing}. 
The initial sampling phase is mostly by random sampling. Additionally, in  Breach with  global Nelder-Mead, 
so-called \emph{corner samples} are  added to the list $\mathsf{U}$. The number of corner samples grows exponentially as $K$ grows, i.e.\ as we have more control points.
\auxproof{
\begin{mydefinition}[signal $\bu_{(u_{1},\dotsc,u_{K})}$]\label{def:piecewiseConstSignal}
 Let $T\in\Rpos$ and $u_{1},\dotsc, u_{K}\in\R^{M}$. A signal $\bu_{(u_{1},\dotsc,u_{K})}\colon [0,T]\to\R^{M}$ is defined as follows. 
 \begin{equation}\label{eq:piecewiseConstSignal}
 \bu_{(u_{1},\dotsc,u_{K})} (t)
\; := \quad
u_{1}\; \text{ (for $t\in [0,\textstyle\frac{T}{K})$)\enspace, }
 \quad\cdots,\quad
u_{K}\; \text{ (for $t\in [\textstyle\frac{(K-1)T}{K},T]$)\enspace. }
 \end{equation}
 In the signal $\bu_{(u_{1},\dotsc,u_{K})}$, potential discontinuities are only found at  the time instants $0,\frac{T}{K}, \frac{2T}{K},\dotsc, \frac{(K-1)T}{K}$. These time instants are called \emph{control points}. 
\end{mydefinition}

It is also common to restrict the value domain of each component of input signals to a certain interval (as opposed to the whole set $\R$). That is, for each $i\in[1,M]$, it is required that $\bu(t)(i)$  be in the interval $[a_{i},b_{i}]$ for any $t\in [0,T]$.

The following is used by Breach with  global Nelder-Mead, as we already discussed. 
\begin{mydefinition}[corner sample]\label{def:cornersample}
Let $a_{i}< b_{i}$ be a pair of real numbers, for each $i\in[1,M]$. Let $T\in\Rpos$ be a time horizon, and $K\in\N$ be the number of control points.  A \emph{corner sample} is a signal 
\begin{math}
 \bu_{(u_{1},\dotsc, u_{K})}\colon [0,T]\longrightarrow [a_{1},b_{1}]\times\cdots\times[a_{M},b_{M}]
\end{math}
such that each $u_{j}\in\R^{M}$ satisfies
$ u_{j}(i)= a_{i}$ or $ u_{j}(i)= b_{i}$, for each $i\in [1,M]$.
\end{mydefinition}
Note that the number of corner samples grows exponentially as $K$ or $M$ grows. 
}

\section{Time Staging in Optimization-Based Falsification}
\label{sec:timeStagedFalsification}
\noindent

\begin{mydefinition}[time-staged deployment of falsification solver]\label{def:timeStagedFalsification}
Let $\mathcal{M}$ be a system model,  $\varphi$ be an $\STL$ formula, and $T\in\Rpos$ be a time horizon. Let $K\in\N$ be a parameter; it is the number of
 time stages.
The \emph{time-staged deployment} of a falsification solver $\Falsify$ is the procedure in Algorithm~\ref{algorithm:timeStagedFalsify}. On the line~\ref{line:jthStage}, the model $\mathcal{M}_{\bu}$ is the continuation of $\mathcal{M}$ after $\bu$ (Def.~\ref{def:continuation}); the score function $\partial_{\bv}\rho_{\varphi}$ is defined by
\begin{equation}\label{eq:scoreFunctionDerivative}
 (\partial_{\bv}\rho_{\varphi})(\bv')
 \;:=\;
 \rho_{\varphi}(\bv\cdot\bv')
 \;\stackrel{(\ref{eq:scoreByRobustness})}{=}\;
 \sem{\bv\cdot\bv',\varphi}\enspace.
\end{equation}
The whole procedure is stochastic (since $\Falsify$ is); an invocation is called a \emph{time-staged falsification trial}. It is \emph{successful} if the returned signal $\bu$ satisfies $\sem{\mathcal{M}(\bu),\varphi}<0$. 
\end{mydefinition}

\begin{algorithm}[t]
\caption{Time-Staged Deployment of a Falsification Solver}
\label{algorithm:timeStagedFalsify}
\begin{algorithmic}[1]
\Require a falsification solver $\Falsify$, a system model $\mathcal{M}$, an $\STL$ formula $\varphi$,  $T\in\Rpos$ and $K\in\N$
\State $\bu\gets ()$  
   \Comment{the input prefix obtained so far. We start with the empty signal $()$}
\For {$j \in \{1,\dotsc,K\}$}
\State $\bu'\gets \Falsify(\mathcal{M}_{\bu},\partial_{\mathcal{M}(\bu)}\rho_{\varphi}, \frac{T}{K})$
   \Comment{synthesizing the $j$-th input segment}
\label{line:jthStage}
\State $\bu\gets \bu\cdot \bu'$ 
   \Comment{concatenate $\bu'$, after which the length of $\bu$ is $\frac{jT}{K}$}
\EndFor
\State \Return $\bu$
 \Comment{a time-staged falsification trial is successful if $\sem{\mathcal{M}(\bu),\varphi} < 0$}
\end{algorithmic}
\end{algorithm}

A falsification trial (i.e.\ an invocation of Algorithm~\ref{algorithm:falsify}) is an iterative process: the more we sample, the more likely we obtain a falsifying input signal.  Since
we run multiple falsification trials
 in Algorithm~\ref{algorithm:timeStagedFalsify} (one trial for each of the  $K$ stages), an important question is how we distribute available time to different stages. 

A simple strategy is to fix the number of samples in each phase of Algorithm~\ref{algorithm:falsify}. Then the predicates $\mathsf{InitialSamplingDone}(T,\mathsf{U})$ and $\mathsf{OptimizationSamplingDone}(\mathcal{M},\rho,T,\mathsf{U})$ are given by $|\mathsf{U}|> N_{\mathrm{max}}^{\mathrm{init}}$ and
$|\mathsf{U}|> N_{\mathrm{max}}^{\mathrm{opt}}$,  where $N_{\mathrm{max}}^{\mathrm{init}},N_{\mathrm{max}}^{\mathrm{opt}}$ are constants. 

An \emph{adaptive} strategy, that we also implemented for the
optimization sampling phase, is to continue sampling until we stop seeing
progress. Here we fix a parameter $N_{\mathrm{max}}^{\mathrm{stuck}}
$, and we stop after $N_{\mathrm{max}}^{\mathrm{stuck}}
$ consecutive samplings without reducing robustness. 
A similar strategy of adaptively choosing the number of samples can be
introduced for random sampling in the initial sampling phase (the
lines~\ref{line:initSamplingStart}--\ref{line:initSamplingEnd} of
Algorithm~\ref{algorithm:falsify}).
\auxproof{For $\mathsf{U}=(\bu_{1},\dotsc,\bu_{|\mathsf{U}|})$, $ \mathsf{OptimizationSamplingDone}(\mathcal{M},\rho,T,\mathsf{U})$ is defined by
\begin{displaymath}
 |\mathsf{U}|\;>\;
 N_{\mathrm{max}}^{\mathrm{stuck}}+
 \textstyle \argmin_{i}\rho(\mathcal{M}(\bu_{i})) 
\enspace.
\end{displaymath}
}

\subsection{Towards Efficient Implementation}
\label{subsec:timeStagedFalsificationResultsForImpl}
A key to  speedup of 
 Algorithm~\ref{algorithm:timeStagedFalsify} is in the line~\ref{line:jthStage}; more specifically, how we handle the previous input prefix $\bu$. Here we discuss two directions, one on the model $\mathcal{M}_{\bu}$ and the other on the score function $\partial_{\bv}\rho_{\varphi}$. (We note that  the suggested enhancements are not currently used in our implementation, because of performance reasons. See below.)

\paragraph{Continuation of Models}
 Optimization-based falsification has a very wide application domain. Since it only requires a \emph{black-box} model $\mathcal{M}$, the concrete form of $\mathcal{M}$ can vary from a program
to a Simulink model
and even a system with  hardware components (\emph{HILS}). 
These models can be very big, and usually
 the bottleneck in falsification lies in \emph{simulation}, that is, to compute $\mathcal{M}(\bu)$ given an input signal $\bu$. 

In the line~\ref{line:jthStage} of Algorithm~\ref{algorithm:timeStagedFalsify}, therefore, 
using the 
definition $\mathcal{M}(\bu\cdot\bu')(T+t)$ in Def.~\ref{def:continuation}
is in principle not a good strategy: it
requires simulation of $\mathcal{M}$ for the whole prefix $\bu\cdot\bu'$, which can be avoided if we can directly simulate the continuation $\mathcal{M}_{\bu}$. In Simulink this is possible by saving the snapshot of the model after a simulation,
via the \verb+SaveFinalState+ model configuration parameter.
In our implementation we do not do so, though, because  the overhead of saving and loading snapshots is currently greater than the cost of simulating. This balance can become different, if we figure out a less expensive way to use snapshots, or if we study more complex models.

\paragraph{Derivative of Formulas}
The situation is similar with the score function $\partial_{\mathcal{M}(\bu)}\rho_{\varphi}$ in
the line~\ref{line:jthStage} of Algorithm~\ref{algorithm:timeStagedFalsify}. Using the presentation $\rho_{\varphi}(\mathcal{M}(\bu)\cdot\bv')$ in~(\ref{eq:scoreFunctionDerivative}) requires scanning the same prefix $\mathcal{M}(\bu)$ repeatedly. 
Desired here is a \emph{syntactic} presentation of $\partial_{\mathcal{M}(\bu)}\rho_{\varphi}$, that will be given as an $\STL$ formula $\partial_{\mathcal{M}(\bu)}\varphi$ such that
$\partial_{\mathcal{M}(\bu)}\rho_{\varphi}=\rho_{(\partial_{\mathcal{M}(\bu)}\varphi)}$. 
This would allow one to utilize available algorithms for computing robustness values $\sem{\bv,\partial_{\mathcal{M}(\bu)}\varphi}$.

\begin{mydefinition}[derivative of flat $\STL$ formulas]\label{def:derivative}
 Let $T\in\Rpos$, and  $\bv\colon [0,T]\to \R^{N}$ be a signal.
 Given an $\STL$ formula $\varphi$ that is \emph{flat} in the sense
 that it does not have nested temporal operators,
 the \emph{derivative} $\partial_{\bv}\varphi$ by $\bv$ is defined inductively as follows.
     \begin{displaymath}
      \begin{array}{l}
       \partial_{\bv}\bigl(f(\vec{x})>0\bigr) 
       \;:\equiv\;
       \mathsf{c}_{\sem{\bv,f(\vec{x})>0}}
       \qquad
       \qquad
       \partial_{\bv} \bot
       \;:\equiv\;
       \bot
       \\
       \partial_{\bv} (\lnot\varphi)
       \;:\equiv\;
       \lnot \partial_{\bv}\varphi
       \qquad
       \partial_{\bv} (\varphi_{1}\land\varphi_{2})
       \;:\equiv\;
       (\partial_{\bv}\varphi_{1})\land(\partial_{\bv}\varphi_{2})
       \\
       \partial_{\bv} (\varphi_1 \UntilOp{I} \varphi_2)
       \;:\equiv\;
       \mathsf{c}_{\sem{\bv,\varphi_1 \UntilOp{I} \varphi_2}}
       \lor
       \bigl(\,(\mathsf{c}_{\sem{\bv,\Box \varphi_{1}}}\land \varphi_{1})\UntilOp{I-T}{\varphi_{2}}\,\bigr)
      \end{array}
    \end{displaymath}
Here the interval $I-T$ is obtained by shifting endpoints, such as $[a,b]-T=[a-T,b-T]$. For each $r\in\R$, the notation $\mathsf{c}_r$ abbreviates the atomic formula $r>0$, where $r$ is thought of as a constant function.  We use the fact that $\sem{\bw,\mathsf{c}_{r}}=r$.
\end{mydefinition}
 Until formulas $\varphi_1 \UntilOp{I} \varphi_2$
 are split into the evaluation on the signal prefix $\bv$ (first disjunct),
 and a continuation (second disjunct).
 The constant $\mathsf{c}_{\sem{\bv,\Box \varphi_{1}}}$ injects
 the robustness of~$\varphi_1$ seen so far into the residual formula
 (recall that both of~$\Box$ and~$\land$ take the infimum).
It follows that
$\partial_{\bv} \bigl(\Box_I~\varphi\bigr)
    \equiv \mathsf{c}_{\sem{\bv,\Box \varphi}} \land \Box_{I-T}~\varphi$
and
$\partial_{\bv} \bigl(\Diamond~\varphi\bigr)
    \equiv \mathsf{c}_{\sem{\bv,\Diamond \varphi}} \lor \Diamond_{I-T}~\varphi$.

 \begin{myproposition}\label{prop:correctnessOfDerivative}
 Let $T\in\Rpos$,  $\bv\colon [0,T]\to \R^{N}$ be a signal, and $\varphi$ be a flat $\STL$ formula. We have, for any
 $\bv'\colon [0,T']\to \R^{N}$,
 \begin{math}
  \sem{\bv',\partial_{\bv}\varphi}
  =
  \sem{\bv\cdot\bv',\varphi}
 \end{math}.
 \myqed
 \end{myproposition}

A proof is in Appendix~\ref{appendix:STLDerivative}. 
Use of derivatives for timed  specifications is also  found e.g.\ in~\cite{DBLP:conf/tacas/UlusFAM16}. The settings are different, though: Boolean semantics in~\cite{DBLP:conf/tacas/UlusFAM16} while our semantics is quantitative. Our restriction to flat formulas comes mainly from this difference, and lifting the flatness restriction seems hard.

\section{Sufficient Conditions for Time Staging}
\label{sec:theoreticalBoundaryCaseResults}

We present some theoretical analyses of the performance of time staging
that indicate to which class of systems the time-staged approach can apply.
We give some \emph{sufficient} conditions under which the approach is guaranteed
to work. However, it should be noted that it is \emph{not} necessary
that a concrete system satisfies these conditions strictly as these are rather restrictive.
Nevertheless, we believe that users with expert domain knowledge
can judge whether their models satisfy these conditions approximately.
This way our results provide those users with ``rules of thumb.''

As we discussed in the last paragraph of~\S{}\ref{sec:preview}, 
the potential performance advantage by time staging comes from the reduction of
search spaces from $|U|^{K}$ to $K\cdot|U|$. Here $U$ is the set of potential
input segments for each stage, and $K$ is the number of stages. This advantage
comes at the risk of missing out some error input signals. The following basic
condition~\eqref{eq:incrementalFalsification},
that we call \emph{incremental falsification}, ensures that there is no such risk.
More precisely, we can decompose the ``best'' input signal~$\bu$ into a first stage~$\bu_1$
and its remainder~$\bu_2$ such that the entire falsification problem (left hand side)
is solved by greedy optimization of the initial segment (inner $\argmin_{\bu_1}$),
and subsequent optimization of the continuation (outer $\min_{\bu_2}$).
For all choices of $T_1,T_2$ with ranges
$\bu\colon [0,T_1 + T_2]\to \Rm$ and $\bu_i\colon [0,T_i]\to \Rm$:
\begin{equation}\label{eq:incrementalFalsification}
\min_{\bu} \sem{\M(\bu),\varphi}
  \; =\; \min_{\bu_2}\Bigl\llbracket\mathcal{M}\Bigl(\bigl( \argmin_{\bu_{1}}\sem{\mathcal{M}(\bu_{1}),\varphi}\bigr)\cdot\bu_{2}\Bigr),\varphi\Bigr\rrbracket
\end{equation}
Algorithm~\ref{algorithm:timeStagedFalsify} repeatedly
unfolds~\eqref{eq:incrementalFalsification} by picking constant $T_1 = T/K$
where~$T$ is the time horizon and~$K$ is the number of stages.
The rest of this section is devoted to the search for concrete sufficient conditions for~(\ref{eq:incrementalFalsification}). 

\paragraph{Monotone Systems and Ceiling Specifications}
We formalize the time monotonicity property
 in~\S{}\ref{sec:preview}. That it implies  incremental falsification~(\ref{eq:incrementalFalsification}) can be easily proved. 

\begin{mydefinition}[time-monotone falsification problem]\label{def:timeMonotonicity}
 A system model $\mathcal{M}$ and an $\STL$ formula $\varphi$ are said to \emph{constitute a time-monotone falsification problem} if, for any input signals
 $\mathbf{u}_{1}, \mathbf{u}'_{1}\colon [0,T_{1}]\to \Rm$ and $\mathbf{u}_{2}\colon [0,T_{2}]\to \Rm$, 
 \begin{math}
   \sem{\mathcal{M}(\mathbf{u}_{1}),\varphi}
  \le
  \sem{\mathcal{M}(\mathbf{u}'_{1}),\varphi}
 \end{math}
 implies
 \begin{math}
  \sem{\mathcal{M}(\mathbf{u}_{1}\cdot\mathbf{u}_{2}),\varphi}
  \le
  \sem{\mathcal{M}(\mathbf{u}'_{1}\cdot \mathbf{u}_{2}),\varphi}
 \end{math}.
\end{mydefinition}

We investigate yet more concrete conditions that ensures time monotonicity. 
The following condition on system models is assumed
in the example of~\S{}\ref{sec:preview}. 

\noindent
 \begin{mydefinition}[monotone system, ceiling specification]
 \label{def:ceiling-specification}
 Let $x$ be a variable (for output).
 A system model $\mathcal{M}$  is said to be \emph{monotone} in 
 $x$ if, for each $\bu_{1},\bu_{1}'\colon[0,T_{1}]\to \R^{M}$ and $\bu_{2}\colon[0,T_{2}]\to \R^{M}$,
 \begin{math}
  \mathcal{M}(\bu_{1})(T_{1})
  (x)
  \le
 \mathcal{M}(\bu_{1}')(T_{1})
  (x)
 \end{math}
 implies
 \begin{math}
  \mathcal{M}(\bu_{1}\cdot\bu_{2})(T_{1}+T_{2})
  (x)
  \le
 \mathcal{M}(\bu_{1}'\cdot\bu_{2})(T_{1}+T_{2})
  (x)
 \end{math}. 

 An $\STL$ formula of the form  $\Box(x<c)$, where $x$ is a variable and $c\in\R$ is a constant, is called a \emph{ceiling formula}. 
 \end{mydefinition}

One can speculate that  a monotone system and a ceiling specification  $\Box(x<c)$, like those in~\S{}\ref{sec:preview}, constitute a time-monotone falsification problem. The speculation  is  not true, unfortunately; a counterexample is easily constructed using a model $\mathcal{M}$ whose output signal is not increasing. We can instead show the following weaker property. 

\begin{mydefinition}[truncated time monotonicity]\label{def:truncatedTimeMonotonicity}
 A system model $\mathcal{M}$ and an $\STL$ formula $\varphi$  \emph{constitute a truncated time-monotone falsification problem} if, for any input
 $\mathbf{u}_{1}, \mathbf{u}'_{1}\colon [0,T_{1}]\to \Rm$ and $\mathbf{u}_{2}\colon [0,T_{2}]\to \Rm$, 
$  \sem{\mathcal{M}(\mathbf{u}_{1}),\varphi}
  \le
  \sem{\mathcal{M}(\mathbf{u}'_{1}),\varphi}
$ implies existence of $T\in(0,T_{1}]$ such that
$  \sem{\mathcal{M}((\mathbf{u}_{1}|_{[0,T]})\cdot\mathbf{u}_{2}),\varphi}
  \le
  \sem{\mathcal{M}((\mathbf{u}'_{1}|_{[0,T]})\cdot \mathbf{u}_{2}),\varphi}\enspace.
$
\end{mydefinition}

\begin{myproposition}
\label{prop:monotoneAndBox}
 Let $\mathcal{M}$ be a model monotone in $x$, and  $\varphi\equiv
\bigl(\Box(x<c)\bigr)$. Then $\mathcal{M}$ and $\varphi$ constitute a truncated time-monotone falsification problem. \myqed
\end{myproposition}

The proof, in Appendix~\ref{appendix:proofOfmonotoneAndBox}, 
 constructs a concrete choice of $T$ in Def.~\ref{def:truncatedTimeMonotonicity}. Specifically it is the instant $T\in[0,T_{1}]$ in which the robustness $\sem{\mathcal{M}(\mathbf{u}_{1}|_{[0,T]}),\varphi}$ is minimum. In the scenario of~\S{}\ref{sec:preview} this is the instant that the vehicle speed is in its peak. 
Note that truncated time monotonicity does not guarantee incremental
falsification as per~\eqref{eq:incrementalFalsification},
but it implies that the current rigid time staging at
$0, \frac{T}{K}, \frac{2T}{K},\dotsc, \frac{(K-1)T}{K}$ is not optimal.
These theoretical considerations suggest potential improvement of the staged
procedure in Def.~\ref{def:timeStagedFalsification} with adaptive choice of
stages, which is left for future work.

\paragraph{Stateless Systems and Reachability Specifications}
Here is another sufficient condition. 

\begin{mydefinition}[stateless system, reachability formula]\label{def:statelessSystem}
  A  system model $\mathcal{M}$ is said to be \emph{stateless} if, 
for any input signals
 $\mathbf{u}_{1},\bu_{1}'\colon [0,T_{1}]\to \Rm$ and $\mathbf{u}_{2}\colon [0,T_{2}]\to \Rm$, we have $\mathcal{M}(\bu_{1}\cdot\bu_{2})|_{(T_{1},T_{2}]}=
\mathcal{M}(\bu'_{1}\cdot\bu_{2})|_{(T_{1},T_{2}]}
$.

An $\STL$ formula  $\Diamond\psi
$, where $\psi$ is modality-free,
 is called a \emph{reachability formula}. 
\end{mydefinition}

Note that being stateless is a sufficient but not necessary condition for 
$
\mathcal{M}(\bu_{1}\cdot\bu_{2})
=
\mathcal{M}(\bu_{1})\cdot
\mathcal{M}(\bu_{2})
$. Statelessness requires insensitivity to previous input prefixes, but a stateless system can still be sensitive to time. 

\begin{myproposition}
\label{prop:statelessAndReachability}
 Let $\mathcal{M}$ be a stateless system and $\varphi$ be a reachability specification  $\varphi\equiv(\Diamond\psi)$. Then $\mathcal{M}$ and $\varphi$ satisfy the incremental falsification property~(\ref{eq:incrementalFalsification}). \myqed
\end{myproposition}

\vspace{.3em}\noindent
A proof is easy. 
A typical situation in which we would appeal to Prop.~\ref{prop:statelessAndReachability} is when: the specification is $\Diamond(x<c_{1}\lor x>c_{2})$ (which can be hard to falsify if $c_{1}<c_{2}$ are close); and the system is already in its stable state (so that its behavior does not depend much on what happened during the transient phase). Our experiments in~\S{}\ref{sec:experiments} demonstrate that time staging can drastically improve performance in such settings. 

\auxproof{
To falsify a reachability specification $\Diamond
\psi
$ is to maintain the  output value so that it satisfies the (non-temporal) specification $\lnot\psi$. A typical example is $\psi\equiv(x<c_{1}\lor x>c_{2})$; to falsify it is to keep the value of $x$  in the interval $[c_{1},c_{2}]$. This can be very difficult if $c_{1}$ and $c_{2}$ are close to each other.

Let us now consider dividing $[0,T]$ into $K$ segments of the same size. Let $\bu_{j}\colon [0,\frac{T}{K}]\to\R^{M}$ be an input segment, for each $j\in[1,K]$. 
\auxproof{We have the following.
\begin{align*}
& \sem{\mathcal{M}(\bu_{1}\cdot\cdots\cdot \bu_{K}),\Diamond(x<c_{1}\lor x>c_{2})}
\\& =
 \sem{\mathcal{M}(\bu_{1})\cdot (\mathcal{M}(\bu_{1}\cdot\bu_{2})|_{[\frac{T}{K},\frac{2T}{K}]})\cdots\cdot (\mathcal{M}(\bu_{1}\cdots\bu_{K})|_{[\frac{(K-1)T}{K},T]}),\,\Diamond(x<c_{1}\lor x>c_{2})}
\\& =
 \bigsqcup_{j\in[1,K]}
 \bigsqcup_{t\in[0,\frac{T}{K}]} 
\sem{
\bigl(\,\mathcal{M}(\bu_{1}\cdots\bu_{j})|_{[\frac{(j-1)T}{K},\frac{jT}{K}]}\,\bigr)^{t}
,x<c_{1}\lor x>c_{2}}
\\& =
 \bigsqcup_{j\in[1,K]}
 \bigsqcup_{t\in[0,\frac{T}{K}]} 
\bigl(\,
c_{1}- (\mathcal{M}(\bu_{1}\cdots\bu_{j})|_{[\frac{(j-1)T}{K},\frac{jT}{K}]})(t)
\,\bigr)
\sqcup
\bigl(\,
 (\mathcal{M}(\bu_{1}\cdots\bu_{j})|_{[\frac{(j-1)T}{K},\frac{jT}{K}]})(t)-c_{2}
\,\bigr)
\end{align*}
Therefore, 
}
That $\bu_{1}\cdot\cdots\cdot \bu_{K}$ is a falsifying input signal means  the following: 
\begin{equation}\label{eq:12082259}
(\mathcal{M}(\bu_{1}\cdots\bu_{j})|_{[\frac{(j-1)T}{K},\frac{jT}{K})})(t) \in\R^{N}
\text{satisfies $\lnot\psi$, for each $j\in[1,K]$ and $t\in[0,\textstyle\frac{T}{K}]$.}
\end{equation}
 Let us assume that we consider only piecewise constant signals with control points $0, \frac{T}{K}, \frac{2T}{K},\dotsc, \frac{(K-1)T}{K}$. 
We further quantize input values, and say that for each stage, $u_{j}$ is chosen from a finite set $U=\{r_{1},\dotsc,r_{|U|}\}$ of real numbers.

In this setting,  the current falsification problem translates to the search problem for $(u_{1},\dotsc,u_{K})\in U^{K}$ that satisfies~(\ref{eq:12082259}). The search space is of size $|U|^{K}$. 

We shall now show how to utilize the statelessness assumption on $\mathcal{M}$ to accelerate search. This goes  in the same time-staged way as in Algorithm~\ref{algorithm:timeStagedFalsify}. 
\begin{itemize}
 \item In the first stage we pick $u_{1}$  so that the corresponding signal segment $\bu_{1}$ makes
\begin{math}
 (\mathcal{M}(\bu_{1}))(t) 
\end{math}
satisfy $\lnot\psi$, for each $t\in[0,\frac{T}{K}]$. There are $|U|$ candidates for such $u_{1}$. 
 \item In the second stage 
we pick $u_{2}\in U$  so that $\bu_{2}$ makes
\begin{math}
 (\mathcal{M}(\bu_{1}\bu_{2}))|_{(\frac{T}{K},\frac{2T}{K}]}(t) 
\end{math}
satisfy $\lnot\psi$, for each $t\in[0,\frac{T}{K}]$. The number of candidates is again $|U|$. 

A point here is that, thanks to statelessness, there is no need of backtracking to the first stage. Indeed, assume that desired $u_{2}$ exists under $u'_{1}$ that is different from $u_{1}$ that we picked before. Then \begin{math}
 (\mathcal{M}(\bu'_{1}\bu_{2}))|_{(\frac{T}{K},\frac{2T}{K}]}(t) 
\end{math}
satisfy $\lnot\psi$, for each $t\in[0,\frac{T}{K}]$. However statelessness of $\mathcal{M}$ asserts
\begin{math}
 (\mathcal{M}(\bu'_{1}\bu_{2}))|_{(\frac{T}{K},\frac{2T}{K}]}
=
(\mathcal{M}(\bu_{1}\bu_{2}))|_{(\frac{T}{K},\frac{2T}{K}]}
\end{math}. Therefore search for $u_{2}$ under $u_{1}$ is as promising as under other prefixes $u'_{1}$. 
 \item We continue until the $K$-th stage. In each stage we search for $u_{j}$ from $|U|$ candidates, and there is no need for backtracking. 
\end{itemize}
Therefore, in this time-staged search, we only have to check $K\cdot|U|$ candidates. This is better than an exhaustive search in the whole input space of size $|U|^{K}$. 

The arguments so far are high-level, and they would
 apply to real  problems only approximately. Indeed, few system
models are stateless (Def.~\ref{def:statelessSystem}); and an
optimization-based falsification solver like Breach usually searches in
the input space $U^{K}$  more efficiently than exhaustive
search. Nevertheless, our experiment results
in~\S{}\ref{sec:experiments} show that time staging improves performance
in those settings which are close to the idealized ones that we just
discussed.
}

\paragraph{Corner Samples for Global Nelder-Mead}
The reduction of search spaces from $|U|^{K}$ to $K\cdot |U|$ has its  analogue 
 in the number of corner
samples in  Breach with
global Nelder-Mead (lines~\ref{line:initSamplingStart}--\ref{line:initSamplingEnd} of
Algorithm~\ref{algorithm:falsify}, see the last paragraph of  \S{}\ref{sec:unstagedFalsification}). Originally the
number of corner samples is $2^{K\cdot M}$, where $K$ is the number of control points and 
 $M$ is the number of input values.
By introducing $K$ time stages, the total number of
corner samples is reduced to $K\cdot 2^{M}$.

\section{Experiments}
\label{sec:experiments}
\newcommand{\yes}{20*}
\newcommand{\no}{0*}

\newcommand{\throttle}{\mathit{throttle}}
\newcommand{\brake}{\mathit{brake}}
\newcommand{\speed}{v}
\newcommand{\rpm}{\omega}
\newcommand{\gear}{g}
\newcommand{\AF}{\mathit{AF}}
\newcommand{\AFref}{\mathit{AF}_\text{ref}}

We compare the success rate and time consumption of the proposed method.
The benchmarks here use automotive Simulink models that are commonly used in the falsification literature. Specifications are chosen taking the deliberations of 
\S{}\ref{sec:theoreticalBoundaryCaseResults} into account,
namely with ceiling specifications
(\Def{ceiling-specification}, including the example of \S{}\ref{sec:preview}),
a reachability specification (\Def{statelessSystem})
and a combination thereof.

The base line is Algorithm~\ref{algorithm:falsify} implemented by Breach~\cite{DBLP:conf/cav/Donze10}.
The methods proposed in~\S{}\ref{sec:timeStagedFalsification}
are implemented on top of Breach:  the time-staged
Algorithm~\ref{algorithm:timeStagedFalsify}~(TS),
and the adaptive strategy (A-TS, the one described after Def.~\ref{def:timeStagedFalsification}).
All three algorithms (plain, TS, A-TS) are combined with different optimization solvers: CMA-ES, simulated annealing (SA), global Nelder-Mead (GNM),
obtaining a total of nine configurations.

The results in \Table{at}
indicate that both success rate and runtime performance
are significantly improved by time staging,
often finding counterexamples when non-staged Breach fails
(e.g.\ columns \emph{S3 hard} and \emph{S init}).
Furthermore, we see that while the adaptive algorithm (A-TS) does not necessarily
lead to a higher success rate in comparison to the time-staged one (TS),
it gives yet another runtime performance improvement.
However, as discussed in detail in \S{}\ref{sec:discussion},
there is no overall best algorithm,
and time staging affects the optimization algorithms differently
depending on the problem.

\begin{table}[t]
\setlength\tabcolsep{3pt}
\centering
\caption{
Experimental results. Each column shows how many falsification trials succeeded (out of 20), and the average runtime.
        %
         S1:~$\Box_{[0,30]}~(\speed < 120)$.
         S2:~$\Box_{[0,30]}~(\gear = 3 \to \speed \ge 30)$.
         S3:~$\Diamond_{[10,30]}~(\speed \le v_\text{min} \lor \speed \ge v_\text{max})$, where:
          $v_\text{min} = 50$, $v_\text{max} = 60$ (easy);
          $v_\text{min} = 53$, $v_\text{max} = 57$ (hard).
         S4:~$\Box_{[0,10]}(v < \speed) \lor \Diamond_{[0,30]}(\rpm > \omega_\text{max})$, where:
          $v_\text{min} = 80, \omega_\text{max} = 4500$ (easy);
          $v_\text{min} = 50, \omega_\text{max} = 2700$ (mid); 
          $v_\text{min} = 50, \omega_\text{max} = 2520$ (hard). 
%
         The specification S for the Abstract Fuel Control model is
         $\neg(\DiaOp{[t_1,t_2]}{\BoxOp{[0,t']}{(\AF - \AFref > \delta * 14.7))}}$, where:
          $t_1 = 0$, $t_2 = 6$, $t' = 3$, $\delta = 0.07$ (init);
          $t_1 = 6$, $t_2 = 26$, $t' = 4$, $\delta = 0.01$ (stable). 
         Starred numbers~0* or~20* indicate that GNM is deterministic
         so all trials yield the same result.}
\label{tab:at}
\scalebox{0.87}{
\begin{tabular}{|r||rr|rr||rr|rr||rr|rr|rr||rr|rr|}
\hline
       model
       & \multicolumn{14}{c||}{Automatic Transmission}
       & \multicolumn{4}{c|}{Abst.\ Fuel Ctrl.}
       \\
\hline
       spec.\
       & \multicolumn{2}{c|}{S1}
       & \multicolumn{2}{c||}{S2}
       & \multicolumn{2}{c|}{S3 easy}
       & \multicolumn{2}{c||}{S3 hard}
       & \multicolumn{2}{c|}{S4 easy}
       & \multicolumn{2}{c|}{S4 mid}
       & \multicolumn{2}{c||}{S4 hard}
       & \multicolumn{2}{c|}{S init}
       & \multicolumn{2}{c|}{S stable}
       \\
\hline
algorithm
  & time & \#/20
  & time & \#/20
  & time & \#/20
  & time & \#/20
  & time & \#/20
  & time & \#/20
  & time & \#/20
  & time & \#/20
  & time & \#/20 \\
\hline
CMA-ES & 27s & \textbf{20}   &  \textbf{5s} & \textbf{20}
       & 39s & 14 & 57s &  0
       & 32s & 16 & 37s &  9 & 59s &  0
       & 49s &  0 & 82s &  1
       \\
   +TS & 52s & 15   & 15s & 16
       & \textbf{9s} & 19 & 23s & \textbf{11}
       & 15s & 14 & \textbf{14s} & \textbf{14} & 24s &  3
       & 30s &0    &  42s & 1
       \\
 +A-TS & 41s & 18   & 15s & 17
       & \textbf{9s} & 16 & 21s & \textbf{10}
       & 26s & 14 & 22s & 14 & \textbf{20s} &  \textbf{5}
       & 26s& 0    & 41s & 0
       \\
\hline
    SA & 50s & 5    & 43s &  7
       & 37s &  9 & 55s &  0
       & 35s &  6 & 36s &  9 & 47s &  \textbf{5}
       & 51s & 0   & 76s & 2
       \\
   +TS & 37s & \textbf{20}   & 33s & \textbf{16}
       & \textbf{11s} & \textbf{19} & 33s &  \textbf{8}
       & 21s & 14 & 25s & 13 & 51s &  0
       & 47s & \textbf{1}   & 54s & \textbf{7}
       \\
 +A-TS & 34s & \textbf{20}   & 18s & \textbf{17}
       &  \textbf{9s} & \textbf{18} & 26s &  \textbf{4}
       & \textbf{16s} & \textbf{18} & 21s & 11 & 30s &  2
       & 34s & 0   & 42s & \textbf{5}
       \\
\hline      
   GNM &\textbf{6s} & \textbf{\yes} & 61s & \no
       & 56s & \no & 55s & \no
       & 43s & \no & 46s & \no & 53s & \no
       & 50s & \no & 86s & \no
       \\
   +TS & 42s & \yes & 15s & \textbf{\yes}
       & \textbf{13s} & \textbf{\yes} & 25s & \textbf{\yes}
       & \textbf{11s} & \textbf{\yes} & 45s & \no & 52s & \no
       & 30s &\textbf{\yes}  & \textbf{20s} &\textbf{\yes}
       \\
 +A-TS & 20s & \yes & 16s & \textbf{\yes}
       & \textbf{10s} & \textbf{\yes} & 26s & \textbf{\yes}
       & 13s & \yes & 45s & \no & 43s & \no
       & 37s & \no & \textbf{19s} & \textbf{\yes}
       \\
\hline
\end{tabular}
}
\end{table}

\paragraph{Benchmarks}
 \emph{Automatic Transmission} is a Simulink model that
was proposed as a benchmark for falsification
in~\cite{DBLP:conf/cpsweek/HoxhaAF14}.
It has input values~$\throttle \in [0,100]$ and~$\brake \in [0,325]$,
and outputs the car's speed~$\speed$, the engine rotation~$\rpm$,
and the selected gear~$\gear$.

With this model we consider five specifications \textbf{S1--5}. The first two are ceiling ones.
Specification \textbf{S1} $\Box_{[0,30]}~(\speed < 120)$
(cf.~the example in \S{}\ref{sec:preview}) states
the speed be always below a threshold.
This property is easily falsified with $\throttle = 100$.
Specification \textbf{S2} $\Box_{[0,30]}~(\gear = 3 \to \speed \ge 30)$
states that it is not possible to drive slowly in a high gear.
A falsifying trajectory first has to speed up to reach this gear
and subsequently roll out until speed falls below the threshold.
This latter part of the trajectory can again be seen as a ceiling specification.
Note that this property is interesting because the robustnes does not provide
any guidance unless gear~3 has been entered by the system.

Specification \textbf{S3} is a reachability problem,
$\Diamond_{[10,30]}~(\speed \le v_\text{min} \lor \speed \ge v_\text{max})$,
that encodes the search for a trajectory that keeps the speed between a lower and upper bound.
The falsification problem consists of two sub-challenges:
1)~hitting this speed interval precisely after an initial acceleration up to 10s simulated time; and
2)~maintaining a correct speed till the time horizon.
This suggests that a natural decomposition of the problem can indeed be achieved
by separating these two aspects in time.

Specification \textbf{S4} $\Box_{[0,10]}(v_\text{min} < \speed) \lor \Diamond_{[0,30]}(\rpm > \omega_\text{max})$
expresses that speed~$v_\text{min}$ can only be reached with an engine rotation exceeding a threshold~$\omega_\text{max}$.
This specification is mentioned in~\cite{DBLP:conf/cpsweek/HoxhaAF14} and evaluated in
e.g., \cite{DBLP:conf/cav/AkazakiH15,DBLP:conf/cav/AdimoolamDDKJ17}, too.
To falsify, a trajectory must be found that reaches speed~$v$
early with an engine rotation lower than~$\omega_\text{max}$.
The difficulty increases with  higher~$v_\text{min}$ and lower~$\omega_\text{max}$, respectively.
The formula represents the mixture of ceiling and reachability specifications.

\smallskip

\noindent
The second system model is \emph{Abstract Fuel Control} from~\cite{DBLP:conf/hybrid/JinDKUB14}.
It takes two input values,
\emph{pedal angle} (from $[0,61.1]$) and \emph{engine speed} (from $[0,1100]$); it outputs 
\emph{air-fuel ratio} $\AF$, which influences fuel efficiency and performance of the car.
The value of $\AF$ is expected to be close to a reference value $\AFref$.
According to~\cite{DBLP:conf/hybrid/JinDKUB14},
this setting corresponds to the so-called \emph{normal mode}, where $\AFref = 14.7$ is constant.

We used the specification $\neg(\DiaOp{[t_1,t_2]}{\BoxOp{[0,t']}{(\AF - \AFref > \delta * 14.7))}}$: the air-fuel ratio does not deviate from an acceptable range 
for more than $t'$ seconds.
 We evaluated this specification with two parameter sets:
the initial period with a larger
error margin, and the stable period with a smaller margin. See Table~\ref{tab:at} for parameter values.

\paragraph{Experimental Setup and Results}
For the experiments with the Automatic Transmission model,
the input signals were piecewise constant with 5 control points.
The  time horizon was $T = 30$.
The parameters outlined in~\S{}\ref{sec:timeStagedFalsification} were as follows:
the maximum number of samplings for each plain (non-staged)
falsification trial
 was
$
150$ (initial and optimization samplings combined).
In the time-staged (TS) trials, we make the number of stages coincide with
that of control points. Analogously, the sampling budget per stage was
set to $
30$ for $K=5$~stages, resulting in overall 150~samplings.
The adaptive algorithm (A-TS) ran with the threshold $N^\mathrm{stall}_{\max} = 30/2 = 15$
per each of five stages.
The experiments with the Abstract Fuel Control model were run up to the time horizon
$T = t_2 + t'$ where $t_{2}$ and $t'$ are as in Table~\ref{tab:at}. 
We used three and five stages, respectively, for the initial and stable specifications. These again coincide with the number of control
points. The TS algorithms conducted 30 samplings in each stage.

The experiments ran Breach version 1.2.9 and MATLAB~R2017b on an Amazon EC2 c4.8xlarge instance with a 36 core Intel(R) Xeon(R) CPU (2.90GHz) and 58G of main memory.
However, we did not use the opportunity to parallelize,
and the time reported is in the same order of magnitude
as that of a modern desktop workstation.

The results are shown in \Table{at}.
They are grouped by the underlying stochastic optimization algorithm:
CMA-ES, simulated annealing (SA) and global Nelder-Mead (GNM).
In each group, we compare
plain (unstaged) Breach to the  time-staged  (TS) and the adaptive time-staged (A-TS) ones.
We compare average runtimes (lower is better)
and the success rate (higher is better),
aggregated over 20 falsification trials for each configuration.
Those good results which deserve attention are highlighted in bold.
Note that the implementation of the global Nelder-Mead algorithm in Breach 
uses a deterministic source of quasi-randomness (Halton sequences),
which implies that whether GNM finds a counterexample
is consistent across all trials (marked with an asterisk $*$).

\paragraph{Discussion}
\label{sec:discussion}
Focusing on the Automatic Transmission model first,
we see that CMA-ES works well for S1,
although GNM performs even better
(6s, supposedly because it uses corner samples, see~\S{}\ref{sec:unstagedFalsification}).
Time staging introduces overhead to CMA-ES and GNM,
because each stage is optimized individually.
In contrast, simulated annealing~(SA) benefits from time staging
for the two ceiling specifications S1--2.
We presume that since~SA emphasizes exploration,
it benefits from  exploitation added by time staging
(cf.~\S{}\ref{sec:preview}).

The second specification S2 is slightly more complex:
before gear~3 is reached, there is no guidance from the robustness semantics,
because $\sem{\place,\gear = 3} = -\infty$ masks any quantitative information on $\speed$.
Hence, falsifying this property needs some luck during the collection of
initial samples in Algorithm~\ref{algorithm:falsify}.
CMA-ES  apparently exploits this, see top of column S2 of \Table{at}.
Considering the other algorithms, SA and GNM, both benefit from time staging:
exploitation of time causality prevents these good trajectory prefixes
from being discarded accidentally once the required gear is reached
(cf.~\Fig{nonstaged}).

The results for S3 are evaluated
with two different choices of parameters.
The harder instance was falsified by the time-staged algorithms only,
which can likely be attributed to the flattening of the search space
from size~$|U|^K$ to $K \cdot |U|$
(\S{}\ref{sec:theoreticalBoundaryCaseResults}).
S4 is evidently harder than the previous ones.
Time staging improves performance in a general tendency
but not in all cases (SA for S4 hard).

The results for the Abstract Fuel Control model (the last two columns in \Table{at})
show that the time-staged algorithms boost the ability to falsify some 
rare events.
The specification for the initial stage where $\AF$ is still unstable (S init) can be considered a rare event
since all the three non-staged algorithms failed to falsify it.
Time-staged SA and time-staged GNM managed to find error inputs.
In the last column (stable period) is remarkable, too,
where success rate and run time of SA and GNM significantly improved.

Overall, while the performance of the non-staged algorithms suffers from tightening the
bounds, the time-staged versions are able to find falsifying trajectories with
good success rates while at the same time exhibiting significantly shorter runtimes.

\section{Related Work}
\label{sec:related}

Falsification is a special case of search-based testing,
so considerable research efforts have been made towards
\emph{coverage}~\cite{DBLP:conf/cav/AdimoolamDDKJ17,DBLP:conf/atva/DeshmukhJKM15,DBLP:conf/formats/KuratkoR14}.
The benefits of coverage in falsification guarantees are twofold.
Firstly, they indicate confidence for correctness in case no counterexamples are found.
Paired with sound robustness estimates for simulations, one can cover a an infinite
parameter space with finitely many simulations.
C2CE~\cite{DBLP:conf/tacas/DuggiralaMVP15} is a recent tool that computes
approximations of reachable states using such an approach.
Secondly, coverage can be utilized for a better balance between  \emph{exploration} and \emph{exploitation}:
 stochastic optimization algorithms can be called in an
interleaved manner, in which coverage guides further exploration.
The approach based on Rapidly-Exploring Random Trees \cite{DBLP:conf/nfm/DreossiDDKJD15}
puts an emphasis on exploration by achieving high coverage of the state space.
In their algorithm, robustness-guided hill-climbing optimization plays a supplementary role.
Compared to these works, our current results go in an orthogonal direction,
by utilizing time causality to enhance exploitation.
The so-called multiple-shooting approach to falsification~\cite{DBLP:conf/emsoft/ZutshiDSK14}
can be seen of a generalization of RRTs.
It consists of: an upper layer that searches for an
abstract error trace given by a succession of cells; and a lower layer
where an abstract error trace is concretized to an actual error trace by picking
points from cells. The approach can discover falsifying traces by backwards
search from a goal region, but needs to concatenate partial traces with
potential gaps, which can fail.
Furthermore it is unclear how to extend it to general STL specifications.
A~survey of simulation based approaches has been done by
Kapinski et al.~\cite{kapinski2016simulation}.

Monotonicity has been exploited in different ways for falsification.
Robust Neighborhood Descent~\cite{abbas2014functional,abbas2015test} (RED)
searches for trajectories incrementally, restarting the search from points of
low robustness. Descent computation of RED assumes explicit derivatives of the
dynamics to guarantee convergence to (local) minima.
It is the same principle underlying Prop.~\ref{prop:monotoneAndBox})
and our experiments indicate that this principle is useful for black-box
optimization, too.
In~\cite{abbas2015test}, RED is paired with simulated annealing to combine local
and global search and to account for more exploration. Doing so for our present
work remains to be done in the future.
In~\cite{DBLP:journals/sttt/HoxhaDF18}, Hoxha et al. mine parameters~$\theta$ under which
specifications~$\phi[\theta]$ are satisfied or falsified by the system.
They show that the robust semantics of formulas
is monotone in~$\theta$ and use that fact to tighten such parameters.
This is orthogonal to this work as it does not use monotonicity of the system itself.
Kim et al.~\cite{DBLP:conf/hybrid/KimAS16} use an idea
similar to~\cite{DBLP:journals/sttt/HoxhaDF18}
to partition specifications into upper bounds and lower ceilings.
However, instead of robustness-guided optimization they use exhaustive
exploration of the input space in a way that in turn requires that the system
dynamics is monotone in the choice of each input at each time point.
This is different from our Def~\ref{def:timeMonotonicity} of time-monotonicity
that aims at incrementally composing good partial choices.

The recent work~\cite{DBLP:conf/nfm/DreossiDS17} introduces a compositional
falsification framework, focusing on those systems which include
machine-learning (ML) components that perform tasks such as image recognition.
While the current work aims at the orthogonal direction of finding rare
counterexamples, we are interested in its combination with the results
in~\cite{DBLP:conf/nfm/DreossiDS17}, given the increasingly important roles of
ML algorithms in CPS.

\section{Conclusions and Future Work}
\label{sec:conclusion}
We have introduced and evaluated the idea of time staging to enhance
falsification for hybrid systems. The proposed method emphasizes exploitation
over exploration as part of stochastic optimization.
As there is no single algorithm that fits every problem
(as a consequence of having no free lunch~\cite{DBLP:journals/tec/DolpertM97}),
having a variety of methods at disposal permits the user of a system
to choose the one suitable for the problem at hand.
We have shown that the proposed approach is a good fit for problems that suitable exhibit
time-causal structures, where it significantly outperforms non-staged algorithms.

Two obvious directions for future work have been pointed out already.
Instead of just picking the best trajectory for each stage,
it might be beneficial to retain a few, potentially diverse ones
in the spirit of evolutionary algorithms (\S{}\ref{sec:preview}).
For example, it would be interesting to explore the space between this work
on one hand and coverage-driven rapidly-exploring random trees.

Another idea is to discover time stages adaptively
(\S{}\ref{sec:theoreticalBoundaryCaseResults},
the discussion after Prop.~\ref{prop:monotoneAndBox}).
For the experiments presented here, we chose to set uniformly fixed stages,
which runs the risk of either being too coarse grained (missing some
falsifying input), or being too fine grained (wasting analysis time).

Finally, another future direction is to explore variations of robust semantics
to mitigate discrete propositions like $\gear = 3$ (\S{}\ref{sec:discussion}),
for example using averaging modalities~\cite{DBLP:conf/cav/AkazakiH15}.
Other t-norms than min/max for the semantics of conjunction/disjunction
could preserve more information from different subformulas.

\paragraph{Acknowledgement.}
This work is supported by ERATO HASUO Metamathematics for Systems Design Project
(No. JPMJER1603), Japan Science and Technology Agency.

\bibliographystyle{eptcs}
\bibliography{references.bib}

\newpage
\appendix

\section{$\STL$ Semantics for Time-Bounded Signals}
\label{appendix:semanticsForTimeBoundedSignals}
\begin{mydefinition}[robust semantics for time-bounded signals]\label{def:semanticsTimeBounded} 
  Let $T\in \Rpos$,  $\bw \colon [0,T] \to \R^{N}$ be a time-bounded signal,
  and $\varphi$ be an $\STL$ formula.
  We define the \emph{robustness} 
  $\Robust{\bw}{\varphi}^{T} \in \R \cup \{\infty,-\infty\}$ of $\bw$ with respect to $\varphi$
as follows, by induction. Here the superscript $T$ is an annotation that designates the time horizon.
    \begin{displaymath}
      \begin{array}{l}
          \Robust{\bw}{f(x_1, \dotsc, x_n) > 0}^{T}  \;\Defeq \;
           f\bigl(\bw(0)(x_1), \dotsc, \bw(0)(x_n)\bigr) 
\qquad
          \Robust{\bw}{\bot}^{T}  \;\Defeq\;  -\infty
\\
          \Robust{\bw}{\neg \varphi}^{T}   \;\Defeq\;   - \Robust{\bw}{\varphi}^{T}\qquad
          \Robust{\bw}{\varphi_1 \wedge \varphi_2}^{T}   \;\Defeq\;   \Robust{\bw}{\varphi_1}^{T} \sqcap \Robust{\bw}{\varphi_2}^{T}\\
          \Robust{\bw}{\varphi_1 \UntilOp{I} \varphi_2}^{T}   \;\Defeq\; 
                                                         \displaystyle{ \Vee{t \in I\cap [0,T]}\bigl(\,\Robust{\bw^t}{\varphi_2}^{T-t} \sqcap 
                                                         \Wedge{t' \in [0, t)} \Robust{\bw^{t'}}{\varphi_1}^{T-t'}\,\bigr)}
      \end{array}
    \end{displaymath}
\end{mydefinition}
The Boolean semantics $\models$, found e.g.\ in~\cite{DBLP:conf/formats/DonzeM10}, allows a similar adaptation to time-bounded signals, too.

\auxproof{
\section{Mealy Machines}\label{appendix:mealymachine}
 The notions and conditions in Def.~\ref{def:systemModel}--\ref{def:continuation} have simpler presentation as \emph{Mealy machines}, in case the notion of time is discrete. Specifically, a system model is a function $\mathcal{M}_{0}\colon S\times \R^{M}\to S\times \R^{N}$ where $S$ is the set of internal states. The function $\mathcal{M}_{0}$, together with a choice $s_{0}\in S$ of an initial state,  induces a function $\mathcal{M}_{0}^{s_{0}}\colon (\R^{M})^{*}\to (\R^{N})^{*}$ as follows: given $u_{1}\dotsc u_{n}\in(\R^{M})^{*}$,  let
 \begin{math}
  s_{i} := \pi_{1}\bigl(\mathcal{M}_{0}(s_{i-1},u_{i})\bigr)
 \end{math}
  and
 \begin{math}
  v_{i} := \pi_{2}\bigl(\mathcal{M}_{0}(s_{i-1},u_{i})\bigr)
 \end{math} for each $i\in[1,n]$.  We then define $  \mathcal{M}_{0}^{s_{0}}(u_{1}\dotsc u_{n}):=v_{1}\dotsc v_{n}$. This function $\mathcal{M}_{0}^{s_{0}}$ is a discrete-time counterpart of $\mathcal{M}$ in Def.~\ref{def:systemModel}. 
 
 For this function $\mathcal{M}_{0}^{s_{0}}$ a causality property like~(\ref{eq:causality}) holds: 
\begin{displaymath}
 \mathcal{M}_{0}^{s_{0}}(u_{1}\dotsc u_{n}u_{n+1}\dotsc u_{n+k})\big|_{[1,n]} = 
  \mathcal{M}_{0}^{s_{0}}(u_{1}\dotsc u_{n}) = v_{1}\dotsc v_{n}\enspace.
\end{displaymath} The continuation $(\mathcal{M}_{0}^{s_{0}})_{u_{1}\dotsc u_{n}}$ is also explicitly presented by $\mathcal{M}_{0}^{s_{n}}$, that is, running $\mathcal{M}_{0}$ with $s_{n}$ as the initial state.
}

\section{Brzozowski Derivative of Flat $\STL$ Formulas}\label{appendix:STLDerivative}
In the time-staged falsification procedure we often encounter the following situation: an $\STL$ formula $\varphi$ and a signal $\bv\colon [0,T]\to \R^{N}$  are fixed; and we have to compute robustness $\sem{\bv\cdot\bv', \varphi}$ for a number of different signals $\bv'\colon [0,T']\to \R^{N}$. To aid such computation, a natural idea is to use a syntactic construct $\partial_{\bv}\varphi$ of \emph{(Brzozowski) derivative}. It should be compatible with robust semantics in the sense that
 \begin{math}
  \sem{\bv\cdot\bv',\varphi}
   =
  \sem{\bv',\partial_{\bv}\varphi}
 \end{math}, reducing the computation of the LHS to that of the RHS. 

Similar use of derivatives is found e.g.\ in~\cite{DBLP:conf/tacas/UlusFAM16}. The settings are different, though: Boolean semantics is used in~\cite{DBLP:conf/tacas/UlusFAM16} while we use quantitative robust semantics. In fact, the definition of derivatives in this section focuses on \emph{flat} formulas (i.e.\ free from nested modalities). This restriction is mandated by the quantitative semantics, as our proof later suggests. Anyway, the definitions and results in this section are new to the best of our knowledge.

We need the following extension of $\STL$ syntax.
\begin{mydefinition}[extended $\STL$]\label{def:extendedSTLAppendix}
 We extend the syntax of $\STL$ (Def.~\ref{def:syntax}) by atomic propositions $\mathsf{c}_{r}$ for each $r\in \R$. The robust semantics in Def.~\ref{def:semantics} (and that in Def.~\ref{def:semanticsTimeBounded} in Appendix~\ref{appendix:semanticsForTimeBoundedSignals}) is extended accordingly: $\sem{\bw,\mathsf{c}_{r}}:=r$. 
\end{mydefinition}
Intuitively $\mathsf{c}_{r}$ is an atomic proposition that  constantly returns the robustness value $r$. 

\begin{mydefinition}[derivative]\label{def:derivativeAppendix}
 Let $T\in\Rpos$, and  $\bv\colon [0,T]\to \R^{N}$ be a time-bounded signal. For each extended $\STL$ formula $\varphi$, we define its \emph{derivative} $\partial_{\bv}\varphi$ by $\bv$ by the following induction. 
     \begin{displaymath}
      \begin{array}{l}
       \partial_{\bv}\bigl(f(\vec{x})>0\bigr) 
       \;:\equiv\;
       \mathsf{c}_{\sem{\bv,f(\vec{x})>0}}
       \qquad
       \partial_{\bv} \mathsf{c}_{r} 
       \;:\equiv\;
       \mathsf{c}_{r}
       \qquad
       \partial_{\bv} \bot
       \;:\equiv\;
       \bot
       \\
       \partial_{\bv} (\lnot\varphi)
       \;:\equiv\;
       \lnot \partial_{\bv}\varphi
       \qquad
       \partial_{\bv} (\varphi_{1}\land\varphi_{2})
       \;:\equiv\;
       (\partial_{\bv}\varphi_{1})\land(\partial_{\bv}\varphi_{2})
       \\
       \partial_{\bv} (\varphi_1 \UntilOp{I} \varphi_2)
       \;:\equiv\;
       \mathsf{c}_{\sem{\bv,\varphi_1 \UntilOp{I} \varphi_2}}
       \lor
       \bigl(\,(\mathsf{c}_{\sem{\bv,\Box \varphi_{1}}}\land \varphi_{1})\UntilOp{I-T}{\varphi_{2}}\,\bigr)
      \end{array}
    \end{displaymath}
 Here the interval $I-T$ is  obtained from $I$ by shifting both of its endpoints earlier by $T$, such as $[a,b]-T=[a-T,b-T]$.
\end{mydefinition}

\begin{mydefinition}[flat $\STL$ formula]\label{def:flatSTLFormulaAppendix}
An $\STL$ formula $\varphi$ is \emph{flat} if it does not have nested temporal modal operators. This means: if $\varphi_{1}\UntilOp{I}\varphi_{2}$ is a subformula of $\varphi$,  then neither $\varphi_{1}$ nor $\varphi_{2}$ contains $\UntilOp{}$. 
\end{mydefinition}
\begin{myproposition}\label{prop:correctnessOfDerivativeAppendix}
 Let $T\in\Rpos$,  $\bv\colon [0,T]\to \R^{N}$ be a signal, and $\varphi$ be a flat $\STL$ formula. We have, for each $T'\in\Rpos$ and $\bv'\colon [0,T']\to \R^{N}$,
 \begin{math}
  \sem{\bv',\partial_{\bv}\varphi}
  =
  \sem{\bv\cdot\bv',\varphi}
 \end{math}.
\end{myproposition}
\begin{proof}
 By induction on the construction of $\varphi$. Most equalities below follow from the definition of $\partial$ and that of $\sem{\place}$. 
 \begin{align*}
  \sem{\bv',\partial_{\bv}\bigl(f(\vec{x})>0\bigr) }
  &
  =
  \sem{\bv',     \mathsf{c}_{\sem{\bv,f(\vec{x})>0}}}
  = \sem{\bv,f(\vec{x})>0} 
  =
  \sem{\bv\cdot\bv',f(\vec{x})>0}
  \\
  \sem{\bv',\partial_{\bv}\mathsf{c}_{r} }
  &
  =
  \sem{\bv', \mathsf{c}_{r}}   
  = r   
  =\sem{\bv\cdot \bv', \mathsf{c}_{r}}
  \\
  \sem{\bv',\partial_{\bv}\bot}
  &=
  \sem{\bv',\bot} = -\infty = \sem{\bv\cdot\bv',\bot}
  \\
  \sem{\bv',\partial_{\bv}(\lnot \varphi)}
  &=
  \sem{\bv',\lnot\partial_{\bv}\varphi} 
  =
  - \sem{\bv',\partial_{\bv}\varphi}
  \stackrel{\text{I.H.}}{=}
  - \sem{\bv\cdot\bv',\varphi}
  =
  \sem{\bv\cdot\bv',\lnot\varphi}
  \\
  \sem{\bv',\partial_{\bv}( \varphi_{1}\land\varphi_{2})}
  &=
  \sem{\bv',(\partial_{\bv} \varphi_{1})\land(\partial_{\bv}\varphi_{2})}
  =
  \sem{\bv',\partial_{\bv} \varphi_{1}}
  \sqcap
  \sem{\bv',\partial_{\bv}\varphi_{2}}
  \\
  &
  \stackrel{\text{I.H.}}{=}
  \sem{\bv\cdot\bv', \varphi_{1}}
  \sqcap
  \sem{\bv\cdot\bv',\varphi_{2}}
  =
  \sem{\bv\cdot\bv', \varphi_{1}\land \varphi_{2}}
 \end{align*}
  Here is a nontrivial case. 
   \begin{align*}
  &\sem{\bv',\partial_{\bv} (\varphi_1 \UntilOp{I} \varphi_2)}
  \\
  &=
  \sem{\bv',
       \mathsf{c}_{\sem{\bv,\varphi_1 \UntilOp{I} \varphi_2}}}
       \sqcup
  \sem{\bv',
       \bigl(\,(\mathsf{c}_{\sem{\bv,\Box \varphi_{1}}}\land \varphi_{1})\UntilOp{I-T}{\varphi_{2}}\,\bigr)
}
  \\
  &=
  \sem{\bv,\varphi_1 \UntilOp{I} \varphi_2}
  \sqcup
  \displaystyle{ 
  \Vee{t \in (I-T)\cap [0,T']}
  \bigl(\,\Robust{\bv'^t}{\varphi_2}
  \sqcap \sem{\bv,\Box \varphi_{1}}
  \sqcap
       \Wedge{t' \in [0, t)} \Robust{\bv'^{t'}}{\varphi_1}\,\bigr)}
  \\
  &=
  \Vee{t\in I\cap [0,T]}\bigl(\sem{\bv^{t},\varphi_{2}}\sqcap \Wedge{t'\in [0,t)} \sem{\bv^{t'},\varphi_{1}}\bigr)
  \\
  &\qquad\sqcup
  \displaystyle{ 
  \Vee{t \in (I-T)\cap [0,T']}
  \bigl(\,\Robust{\bv'^t}{\varphi_2}
  \sqcap 
  \bigl(\,\Wedge{t'\in [0,T]}\sem{\bv^{t'},\varphi_{1}}\,\bigr)
  \sqcap
 \Wedge{t' \in [T, T+t)} \Robust{(\bv\cdot\bv')^{t'}}{\varphi_1}\,\bigr)}
  \\
  &\stackrel{(*)}{=}
  \Vee{t\in I\cap [0,T]}\bigl(\sem{(\bv\cdot \bv')^{t},\varphi_{2}}\sqcap \Wedge{t'\in [0,t)} \sem{(\bv\cdot \bv')^{t'},\varphi_{1}}\bigr)
  \\
  &\qquad\sqcup
  \displaystyle{ 
  \Vee{t'' \in I\cap [T,T+T']}
  \bigl(\,\Robust{(\bv\cdot\bv')^{t''}}{\varphi_2}
  \sqcap 
  \bigl(\,\Wedge{t'\in [0,T]}\sem{(\bv\cdot\bv')^{t'},\varphi_{1}}\,\bigr)
  \sqcap
 \Wedge{t' \in [T, t'')} \Robust{(\bv\cdot\bv')^{t'}}{\varphi_1}\,\bigr)}
  \\
  &=
  \Vee{t\in I\cap [0,T]}\bigl(\sem{(\bv\cdot \bv')^{t},\varphi_{2}}\sqcap \Wedge{t'\in [0,t)} \sem{(\bv\cdot \bv')^{t'},\varphi_{1}}\bigr)
  \\
  &\qquad\sqcup
  \displaystyle{ 
  \Vee{t'' \in I\cap [T,T+T']}
  \bigl(\,\Robust{(\bv\cdot\bv')^{t''}}{\varphi_2}
  \sqcap 
  \Wedge{t'\in [0,t'')}\sem{(\bv\cdot\bv')^{t'},\varphi_{1}}
\,\bigr)}
  \\
  &=
   \Vee{t\in I\cap [0,T+T'']}\bigl(\sem{(\bv\cdot \bv')^{t},\varphi_{2}}\sqcap \Wedge{t'\in [0,t)} \sem{(\bv\cdot \bv')^{t'},\varphi_{1}}\bigr)
  \\
  &=
  \sem{\bv\cdot\bv', \varphi_{1}\UntilOp{I}\varphi_{2}}
 \end{align*}
 In $(*)$ we used the following facts. Firstly, for a formula $\psi$ without temporal operators, we have $\sem{\bv,\psi}=\sem{\bv\cdot\bv',\psi}$.   Secondly, if $\bv$'s domain is $[0,T]$ and $t\in [0,T]$, then $\bv^{t}\cdot \bv'=(\bv\cdot\bv')^{t}$.   \myqed
\end{proof}
Note that the flatness assumption on $\varphi$ is crucially used in the proof step. Modifying Def.~\ref{def:derivative} in order to accommodate nested modalities seems hard, after analyzing the proof step $(*)$.

\section{Omitted Proofs}
\subsection{Proof of Prop.~\ref{prop:monotoneAndBox}}\label{appendix:proofOfmonotoneAndBox}
\begin{proof}
By the definitions we have, for each input signal $\bu\colon[0,T]\to\R^{M}$,
\begin{align*}
  \sem{\mathcal{M}(\bu),\Box(x<c)}
  &=
  \Wedge{t\in[0,T]}c-
\mathcal{M}(\bu)(t)(x)
\enspace.
\end{align*}
Therefore the assumption 
\begin{math}
   \sem{\mathcal{M}(\mathbf{u}_{1}),\varphi}
  \le
  \sem{\mathcal{M}(\mathbf{u}'_{1}),\varphi}
\end{math}
expands to 
\begin{equation}\label{eq:2017-12-08-20-27}
\begin{aligned}
   \Wedge{t\in[0,T_{1}]}c-
\mathcal{M}(\bu_{1})(t)(x)
&  \le
  \Wedge{t\in[0,T_{1}]}c-
\mathcal{M}(\bu_{1}')(t)(x)
\enspace.
\end{aligned}
\end{equation}
The first infimum in the above is taken over a compact domain $[0,T_{1}]$; therefore there exists $T\in [0,T_{1}]$ that achieves the infimum. Let $T$ be such a real number. The following is obvious.
\begin{equation}
\begin{aligned}
&   \Wedge{t\in[0,T_{1}]}c-
\mathcal{M}(\bu_{1})(t)(x)
 =
  \Wedge{t\in[0,T]}c-
\mathcal{M}(\bu_{1})(t)(x)
 = 
 c-
\mathcal{M}(\bu_{1})(T)(x)
\\&
\le
  \Wedge{t\in[0,T_{1}]}c-
\mathcal{M}(\bu_{1}')(t)(x)
\le
  \Wedge{t\in[0,T]}c-
\mathcal{M}(\bu_{1}')(t)(x)
\enspace.
\end{aligned}
\end{equation}
Another immediate consequence, derived using the causality of $\mathcal{M}$~(Def.~\ref{def:systemModel}), is
\begin{equation}\label{eq:2017-12-08-21-43}
 c-
\mathcal{M}(\bu_{1}|_{[0,T]})(T)(x)
 \;\le \;
 c-
\mathcal{M}(\bu'_{1}|_{[0,T]})(T)(x)\enspace.
\end{equation}
Our goal is to show
\begin{math}
   \sem{\mathcal{M}(\mathbf{u}_{1}|_{[0,T]}\cdot\mathbf{u}_{2}),\varphi}
  \le
  \sem{\mathcal{M}(\mathbf{u}'_{1}|_{[0,T]}\cdot \mathbf{u}_{2}),\varphi}
\end{math}.
\begin{align*}
&   \sem{\mathcal{M}(\mathbf{u}'_{1}|_{[0,T]}\cdot \mathbf{u}_{2}),\Box(x<c)}
\\
  &=
  \Wedge{t\in[0,T+T_{2}]}c-
\mathcal{M}(\mathbf{u}'_{1}|_{[0,T]}\cdot \mathbf{u}_{2})(t)(x)
\\
  &=
  \Wedge{t\in[0,T]}c-
  \mathcal{M}(\mathbf{u}'_{1})(t)(x)
 \sqcap
  \Wedge{t\in(T,T+T_{2}]}c-
   \mathcal{M}(\mathbf{u}'_{1}|_{[0,T]}\cdot \mathbf{u}_{2}|_{[0,t-T]})(t)(x)
\tag*{$(*)$}
\\
  &\ge
  \Wedge{t\in[0,T]}c-
   \mathcal{M}(\mathbf{u}_{1})(t)(x)
\sqcap
  \Wedge{t\in(T,T+T_{2}]}c-
  \mathcal{M}(\mathbf{u}'_{1}|_{[0,T]}\cdot \mathbf{u}_{2}|_{[0,t-T]})(t)(x)
\quad\text{by~(\ref{eq:2017-12-08-20-27})}
\\
  &\ge
  \Wedge{t\in[0,T]}c-
   \mathcal{M}(\mathbf{u}_{1})(t)(x)
 \sqcap
  \Wedge{t\in(T,T+T_{2}]}c-
   \mathcal{M}(\mathbf{u}_{1}|_{[0,T]}\cdot \mathbf{u}_{2}|_{[0,t-T]})(t)(x)
\tag*{$(\dagger)$}
\\
&=\cdots 
\\
&=   \sem{\mathcal{M}(\mathbf{u}_{1}|_{[0,T]}\cdot \mathbf{u}_{2}),\Box(x<c)}.
\end{align*}
In the above we heavily used the causality of $\mathcal{M}$~(Def.~\ref{def:systemModel}). For example, 
in the step $(*)$ above, causality is used in deriving $\mathcal{M}(\bu_{1}'|_{[0,T]}\cdot \bu_{2})(t)=\mathcal{M}(\bu_{1}')(t)$. In the step $(\dagger)$ we applied the monotonicity of $\mathcal{M}$ to the signals $\bu_{1}|_{[0,T]}$, $\bu'_{1}|_{[0,T]}$ and $\bu_{2}|_{[0,t-T]}$. Note that (\ref{eq:2017-12-08-21-43}) allows to do so. \myqed
\end{proof}

\end{document}